\documentclass[11pt]{article}
\usepackage[margin=1in]{geometry}
\usepackage[english]{babel}
\usepackage{xr}

\usepackage{setspace}
\usepackage{tikz}
\usetikzlibrary{arrows,automata,positioning,arrows.meta,calc}

\usepackage{mathrsfs,hyperref, algorithm, algpseudocode}
\usepackage[]{amsmath}
\usepackage{amsthm}
\usepackage{fix-cm}
\usepackage[]{amssymb}
\usepackage[]{latexsym}
\usepackage[right]{eurosym}
\usepackage[T1]{fontenc}
\usepackage[]{graphicx}
\usepackage[]{epsfig}
\usepackage{fancyhdr}
\usepackage{bbm}
\usepackage{pstricks}
\usepackage{multirow}
\usepackage{natbib}
\usepackage{subcaption}

\allowdisplaybreaks

\makeatletter

\newcommand{\bbr}{\mathbb{R}}
\newcommand{\bbd}{\mathbb{D}}
\newcommand{\bbi}{\mathbb{I}}
\newcommand{\E}{\mathbb{E}}

\renewcommand{\P}{\mathbb{P}}

\newcommand{\bbt}{\mathbb{T}}
\renewcommand{\vec}[1]{\mathbf{#1}}

\newcommand{\vx}{\vec{x}}
\newcommand{\vL}{\vec{L}}
\newcommand{\vp}{\vec{p}}
\newcommand{\vP}{\vec{P}}
\newcommand{\vK}{\vec{K}}
\newcommand{\vV}{\vec{V}}
\newcommand{\vR}{\vec{R}}
\newcommand{\vF}{\vec{F}}
\newcommand{\va}{\boldsymbol{\alpha}}
\newcommand{\vt}{\boldsymbol{\tau}}
\newcommand{\vi}{\boldsymbol{\iota}}

\newcommand{\fcal}{\mathcal{F}}

\newcommand{\lcal}{\mathcal{L}}

\newcommand{\dt}{\Delta t}

\renewcommand{\succ}{{\cal S}}

\newcounter{modcount}
\newcommand{\modulo}[2]{%
\setcounter{modcount}{#1}\relax
\ifnum\value{modcount}<#2\relax
\else\relax
\addtocounter{modcount}{-#2}\relax
\modulo{\value{modcount}}{#2}\relax
\fi}
\newcommand{\tablepictures}[4][c]{\begin{tabular}[#1]{@{}c@{}}#2\vspace{0.5cm}\\(\alph{#4}) #3\end{tabular}}
\newcounter{gridsearch}
\newcommand{\tabpic}[2]{
    \stepcounter{gridsearch}
    \modulo{\thegridsearch}{2}
    \ifnum\value{modcount}=0
        \tablepictures[t]{#1}{#2}{gridsearch}\\[2.0cm]
    \else
        \tablepictures[t]{#1}{#2}{gridsearch}&~&
    \fi
}

\makeatother
\newtheorem{lemma}{Lemma}[section]
\newtheorem{proposition}[lemma]{Proposition}
\newtheorem{theorem}[lemma]{Theorem}
\newtheorem{corollary}[lemma]{Corollary}

\newtheorem{example1}[lemma]{Example}
\newtheorem{rem1}[lemma]{Remark}

\newtheorem{alg1}[lemma]{Algorithm}
\newtheorem{me1}[lemma]{Mechanism}

\newenvironment{remark}{\begin{rem1}\rm}{\end{rem1}}
\newenvironment{example}{\begin{example1}\rm}{\end{example1}}

\usepackage{color}
\usepackage{colortbl}

\newcommand{\T}{\top}
\newcommand{\diag}[1]{\operatorname{diag}(#1)}

\DeclareMathOperator*{\FIX}{FIX}

\newcommand\ind[1]{\mathbbm{1}_{\{#1\}}}

\begin{document}

\title{\vspace{-20pt}Endogenous distress contagion in a dynamic interbank model: how possible future losses may spell doom today
 }
\author{
Zachary Feinstein\thanks{Stevens Institute of Technology, School of Business, Hoboken, USA. {\tt zfeinste@stevens.edu}.} 
\and 
Andreas S{\o}jmark\thanks{London School of Economics, Department of Statistics, London, UK. {\tt a.sojmark@lse.ac.uk}.}
}
\date{\today}
\maketitle
\abstract{
We introduce a dynamic and stochastic interbank model with an endogenous notion of distress contagion, arising from rational worries about future defaults and ensuing losses. This entails a mark-to-market valuation adjustment for interbank claims, leading to a forward-backward approach to the equilibrium dynamics whereby future default probabilities are needed to determine today’s balance sheets. Distinct from earlier models, the resulting distress contagion acts, endogenously, as a stochastic volatility term that exhibits clustering and down-market spikes. Furthermore, by incorporating multiple maturities, we provide a novel framework for constructing systemic interbank term structures, reflecting the intertemporal risk of contagion. We present the analysis in two parts: first, the simpler single maturity setting that extends the classical interbank network literature and, then, the multiple maturity setting for which we can examine how systemic risk materialises in the shape of the resulting term structures.\\[4pt]
}
\\[28pt]
\noindent {\bf Correspondence}\\Andreas S{\o}jmark, Department of Statistics, London School of Economics,  69 Aldwych, London, WC2B 4RR, UK. Email: a.sojmark@lse.ac.uk.
\\[8pt]
\noindent  {\bf Funding information} \\ Part of this work was funded by the OeNB anniversary fund, project number 17793.

\newpage

\section{Introduction}\label{sec:intro}

Since the 2008 financial crisis, it is widely accepted that systemic risk can cause outsized losses within the financial system due to complex feedback mechanisms. Arguably, the most classical channel for this is direct solvency contagion at default events. In this paper, we study instead a dynamic form of distress contagion that can arise endogenously, today, due to banks' (rational) expectations about potential future defaults and contagion within the network.

The fact that systemic risk remains a very real issue was highlighted by the significant stress to the global financial system in the run-up to and following the recent collapse of Silicon Valley Bank (SVB). Beyond a period of sharp increases in volatility and ailing stock prices along with surging CDS and interbank funding spreads (\cite{IMF_2024}), we saw the outright failings of Silvergate Bank, Signature Bank, and First Republic, as well as the global banking behemoth Credit Suisse. Concerning the latter, Andrew Kenningham, Chief European Economist at Capital Economics, noted that:~\emph{``There's sort of contagion---not because they [SVB] had any connection with Credit Suisse, but because investor sentiment changed and people were scrutinising other banks more carefully''}.\footnote{Quote taken from the article ``How SVB Triggered Credit Suisse’s Latest Mess—and Sparked Fears of a Financial Crisis'', Barron's, March 15, 2023} As Credit Suisse was already suffering, they may have been more susceptible to concerns about the future states of the financial system.

In the Fed's report on SVB, Michael S.\ Barr (\cite{Barr_Fed}), Vice Chair for Supervision at the Fed, stressed that \emph{``contagion from the firm’s failure posed systemic consequences not contemplated by the Federal Reserve’s tailoring framework''}. To capture the potential for such contagion, three main points will be at the center of the framework we develop in this paper. Firstly, the above are all instances of early defaults (e.g., due to capital requirements) in a dynamic setting. Secondly, \cite{Barr_Fed} highlights that:~\emph{``This experience has emphasized why strong bank capital matters. While the proximate cause of SVB’s failure was a liquidity run, the underlying issue was concern about its solvency''}. In our model, capital is the key quantity, and dynamic concerns about future solvency (as determined by the capital) will drive a form of distress contagion in the present. Finally, \cite{Barr_Fed} notes that \emph{``concerns about one firm spread to other firms---even if the firm is not extremely large, highly connected to other financial counterparties, or involved in critical financial services''}. This relates closely to the `crisis of confidence' interpretation of interbank contagion advocated by \cite{glasserman2016contagion}. To capture losses of confidence, they incorporate a mark-to-market valuation of interbank claims into the (static and deterministic) Eisenberg--Noe framework through an ad-hoc valuation function (see Appendix \ref{sect:existing_works} for details). Here, we propose a canonical mechanism for such a mark-to-market valuation in a dynamic and stochastic equilibrium model with early defaults and multiple maturities.

\subsection{Primary contributions}

The central contribution of this paper is the construction of a dynamic interbank model with a mark-to-market valuation of interbank claims that accounts, dynamically, for potential future default contagion within the network, measured by the banks' interconnected conditional (risk-neutral) probabilities of default. First, we shall treat the case of a single maturity in Section~\ref{sec:1maturity}. Then, the extension to multiple maturities is the topic of Section \ref{sec:Kmaturity}. Numerical case studies are presented throughout to demonstrate the financial implications.

Our mark-to-market accounting of interbank assets is in contrast to prior works on dynamic default contagion in interbank networks which---implicitly or explicitly---consider historical price accounting, i.e., all debts are assumed to be paid in full until a default is realized (causing a downward jump in capital). For a brief literature review, see Section \ref{sec:intro-lit}. Focusing on systemic risk, mark-to-market accounting will generally yield a worse outcome for the financial system, as compared to historical price accounting. From a regulatory perspective, if the system is healthy after a valuation adjustment based on  our mark-to-market model, it is therefore reasonable to consider it resilient to shocks, as this represents a worst-case scenario for standard accounting rules. Within our framework, this comparison is invariably true in the single maturity setting (see Appendix \ref{sec:hpa-1}). In the multiple maturity setting, however, we have to account for the impact of defaults on liquidity and rebalancing of assets which complicates the picture. Nevertheless, the comparison remains valid globally in time if we assume zero recovery and take the rebalancing strategies to be independent of bank performance (see Appendix \ref{sec:hpa-K}).

As we will see, a key consequence of marking-to-market is that it creates an avenue for interconnectedness to endogenously affect the volatility of equity returns (beyond the exogenous correlation structure of the banks' assets). In this regard, the so-called `leverage effect' and the related notion of `volatility clustering' are two important stylized facts (\cite{cont_stylized}) that remain debated in the financial literature. Despite the name, \cite{Hasanhodzic_Lo} show that the `leverage effect' is very much the norm even amongst unlevered firms. Moreover, \cite{Figlewski_Wang} conclude that it \emph{``is really a `down-market effect’ that may have little direct connection to firm leverage''}, stressing also that there is \emph{``striking asymmetry of the `leverage effect' between up and down markets''} with \emph{``a strong impact on volatility when stock prices fall and a much weaker effect, or none at all, when they rise''}. Our model provides a new take on this, highlighting how systemic risk may naturally materialise in volatility clustering and a `down market effect'. This arises from dynamically adjusting levels of distress contagion, regulated by an additional endogenous stochastic component to the volatility of bank equity determined in equilibrium, accurately reflecting the current probability of future solvency for all members of the system. Due to network effects, the spikes in volatility impact firms regardless of the riskiness of their book, so long as their counterparties have a chance of defaulting.

Finally, by incorporating multiple maturities, our framework leads to an endogenous interbank term structure, whose shape reflects the dynamic nature and severity of systemic risk. Recently, there has been a significant interest in the possibility of predicting \emph{future} financial crises from current data (see \cite{greenwood2022predictable, Richter_2021}). In relation to this, \cite{Bluwstein_2023} (see also \cite{babecky2014banking}) document how the slope of yield curves can have strong predictive power, emphasizing a materially higher risk of a financial crisis when yield curves are negatively sloped even after controlling for recessions (for which yield curves have well-documented predictive power). Holding all else fixed, we demonstrate how an inverted shape of the yield curve for one institution can arise solely from the riskiness of other banks in the system, due to the imminent potential for loss-of-confidence driven contagion. As far as the authors are aware, our proposed multiple maturity model is the first that incorporates a notion of interbank term structures directly within the modeling of systemic risk.

\subsection{Related literature on interbank networks}\label{sec:intro-lit}

The network of interbank obligations, and the resulting default contagion, has been studied in the seminal works of~\cite{EN01, GK10, RV13} in deterministic one-period systems. These static systems have been extended in a number of directions and this remains an active area of research. We refer the interested reader to~\cite{glasserman2016contagion,AW_15} for surveys of the literature. 

Dynamic extensions of the works of~\cite{EN01,RV13,GK10} have been studied in~\cite{CC15,ferrara16,KV16} in discrete-time settings and~\cite{Lipton2016, Lipton2017, sonin2020continuous, feinstein2021dynamic, BBF18} in continuous-time frameworks. All of these cited works, however, implicitly use a historical price accounting rule for interbank assets. That is, all interbank assets are either marked as if the payment will be made in full (prior to maturity or a default) or marked based on the realized payment (after debt maturity or a default).  This is in contrast to mark-to-market accounting (generally used for tradable assets) whereby the value of interbank assets would depend on the expected payments prior to maturity or a default. 
Such re-marking of interbank assets, ahead of or without actual defaults, leads to a form of distress contagion. Distress contagion in interbank networks is studied in the work of~\cite{veraart2020distress} within an Eisenberg--Noe framework. Due to the static nature, the distress contagion imposed in that work comes from exogenous rules, while our goal is to endogenize the distress contagion by linking it to the current conditional probabilities of future defaults in a dynamic setting. The fact that we rely on a mark-to-market mechanism is reminiscent of the literature on indirect, price-mediated contagion, see in particular \cite{CS16} and \cite{CW13, CW14}. In those works, however, the mark-to-market losses result from fire sales in stress test scenarios, specifically in relation to the selling of common holdings, while we develop an equilibrium model for the endogenous valuation of interbank obligations without any selling or related actions in response to a shock.

A natural point of view for our model is that of a central bank or other financial actor wishing to perform a tractable valuation adjustment to understand, today, the full potential for informational contagion within the network going forward, in case (rational) worries about future insolvency and illiquidity were to set in motion a spiral of devaluations. We thus contribute to the literature on network valuation adjustments (e.g.,~\cite{CS07network, fischer2014,bardoscia_2019, barucca2020network,banerjee2022pricing}) which has focused on a single maturity and which has so far not considered a fully dynamic framework.

Arguably, the work of \cite{bardoscia_2019} comes closest to the framework we develop here. Similarly to our model, in the single maturity setting, they consider a forward-looking notion of contagion with early defaults. However, their default rule only involves the realized capital at the initial time together with exogenous dynamics of the external assets for each bank on its own. Thus, very differently from our framework, asset correlations are irrelevant and the equilibrium valuations are only solved for at the initial time, ignoring dynamic interactions. In fact, their model can be seen to reduce to the static one-period framework of \cite{glasserman2016contagion} with a particular justification for the choice of valuation functions. Details of this are given in Appendix~\ref{sect:existing_works}. Finally, we wish to stress that the inclusion of multiple maturities is new and that, unlike earlier works, this requires a treatment of not only insolvency but also illiquidity for the endogenous determination of defaults, both present and future.

\section{The single maturity setting}\label{sec:1maturity}

In this section, we let all interbank and external obligations have the same given maturity $T$. This allows for a simpler presentation than the multiple maturity setting of Section~\ref{sec:Kmaturity} and it separates out those effects that arise already with a single maturity, due to the endogenous notion of defaults and the dynamic valuation adjustment for the interbank network.

In Sections \ref{sec:setting-network} and \ref{sec:1maturity-bs}, we introduce our model. In Section \ref{sec:1maturity-model} we establish the existence of clearing solutions and provide a characterization in terms of stochastic volatility. In Section \ref{sect:fwd_bwd_dynprog}, we derive useful properties and equivalent representations of the clearing solutions. To have a directly implementable model, we focus our analysis on a tractable multinomial tree structure for the randomness in the financial system (see Section \ref{sec:setting-tree}). Based on this, Section \ref{sec:1maturity-cs} presents numerical case studies that highlight key features of the model.

\begin{remark}As briefly discussed above, we note that the related work of \cite{bardoscia_2019} initially considers a continuous time setting, but their model then simplifies to a static problem (see Appendix~\ref{sect:existing_works}). In particular, there is no system of stochastic processes to simulate and therefore the relevance of a tractable tree structure is not a concern in that work.
\end{remark}

\subsection{A static benchmark model of default contagion}\label{sec:setting-network}

To briefly provide the financial context and to fix ideas, we first discuss a static setting with default contagion due to less than full recovery at defaults, akin to \cite{GK10}. Consider a financial system comprised of $n$ banks (or other financial institutions) labeled $i=1,2,...,n$. The balance sheet of each bank is made up of both interbank and external assets and liabilities.
Specifically, on the asset side of the balance sheet, bank $i$ holds external assets $x_i \geq 0$ and interbank assets $L_{ji} \geq 0$ for each potential counterparty $j \neq i$ (in this case, $L_{ii} = 0$ so as to avoid self-dealing).
On the other side of the balance sheet, bank $i$ has liabilities $\bar p_i := \sum_{j = 1}^n L_{ij} + L_{i0}$ with external liabilities $L_{i0} \geq 0$.
We will often denote $\vx := (x_1,x_2,...,x_n)^\T \in \bbr^n_+$, $\vL := (L_{ij})_{i,j = 1,...,n} \in \bbr^{n \times n}_+$, $\vL_0 := (L_{ij})_{i = 1,...,n;~j=0,1,...,n} \in \bbr^{n \times (n+1)}_+$, and $\bar\vp := \vL_0\vec{1} \in \bbr^n_+$.

We can now formulate the default contagion problem as follows.  Let $P_i \in \{0,1\}$ be the indicator of whether bank $i$ is solvent ($P_i = 1$) or in default on its obligations ($P_i = 0$).  Following a notion of recovery of liabilities, if a bank is in default, it will repay a fraction $\beta \in [0,1]$ of its obligations. We note that the Rogers--Veraart model~(\cite{RV13}) is comparable, but with recovery of assets instead.
Mathematically, $\vP \in \{0,1\}^n$ solves the fixed point problem
\begin{equation}\label{eq:GK}
\vP = \psi(\vP) := \ind{\vx + \vL^\T [\vP + \beta(\vec{1}-\vP)] \geq \bar\vp} = \ind{\vx + \vL^\T[\beta + (1-\beta)\vP] \geq \bar\vp},
\end{equation}
and Tarski's fixed point theorem yields the following result about its solutions.
\begin{proposition}\label{prop:GK}
The set of clearing solutions to~\eqref{eq:GK}, i.e.,  $\{\vP^* \in \{0,1\}^n \; | \; \vP^* = \psi(\vP^*)\}$, forms a lattice in $\{0,1\}^n$ with greatest and least solutions $\vP^\uparrow \geq \vP^\downarrow$.
\end{proposition}

To demonstrate the generic non-uniqueness of the clearing solutions, we provide a very simple financial system that admits two distinct solutions. We have set this up so that it can serve as the foundation for various later examples aimed at illustrating the more nuanced aspects of the framework we introduce below, both in the single and multiple maturity settings.
\begin{example}\label{ex:running-static}
Consider the $n = 2$ bank system with external assets $\vx = (1.9,1.5)^\T$, interbank assets/liabilities $L_{ij} = \ind{i \neq j}$ for $i,j \in \{1,2\}$, and external liabilities $L_{i0} = 1$ for $i \in \{1,2\}$. As in~\cite{GK10}, we set the recovery rate $\beta = 0$.
Then, \eqref{eq:GK} admits the following two clearing solutions $\vP^\uparrow = \vec{1}$ and $\vP^\downarrow = \vec{0}$.
\end{example}

To conclude our discussion of the static system, let $\vP^*$ be an arbitrary clearing solution of~\eqref{eq:GK}.  The resulting net worths $\vK^* = \vx + \vL^\T [\beta + (1-\beta)\vP^*] - \bar\vp$ then provide the difference between realized assets and liabilities.  The cash account $\vV^*$, given by $ \vV^*= (\vK^*)^+$, provides the assets-on-hand for each institution immediately after liabilities are paid; notably, the cash account equals the net worths if, and only if, the bank is solvent, otherwise it is zero.
When measuring systemic risk, often the value of payments to the external (or societal) node is desired; herein that value is provided by $K_0^* = V_0^* = \sum_{i = 1}^n L_{i0} [\beta + (1-\beta)P_i^*]$.

\begin{remark}
As mentioned above, throughout this work we consider the case of recovery of liabilities. This notion corresponds to the `recovery of face value' accounting rule in the corporate bond literature (see, e.g.,~\cite{guo2008distressed,hilscher2021valuation}). The recovery of assets case, as in~\cite{RV13}, could be considered instead. We restrict ourselves to the recovery of liabilities because this formulation has been shown to \emph{``provide a better approximation to realized recovery rates''} (\cite{hilscher2021valuation}) compared to the recovery of assets formulation.
\end{remark}

\subsection{The dynamic model and its balance sheet construction}\label{sec:1maturity-bs}

Departing from the static setting of the previous subsection, we now consider the problem of deriving an appropriate \emph{dynamic} model with a mark-to-market valuation of interbank claims.

Consider, first, the banking book for bank $i$. As in Section~\ref{sec:setting-network}, the bank holds two types of assets: \emph{interbank assets} $\sum_{j = 1}^n L_{ji}$ where $L_{ji} \geq 0$ is the total obliged from bank $j$ to $i$, and  \emph{external assets} $x_i$. Furthermore, the bank has liabilities $\bar p_i = \sum_{j = 1}^n L_{ij} + L_{i0}$ due at a future time $T$, where $L_{i0} \geq 0$ is the \emph{external obligations} of bank $i$. The external assets are held in liquid and marketable assets so that, distinct from Section~\ref{sec:setting-network}, their current value $x_i(t)$ fluctuates over time. We thus model $x_i=(x_i(t))_{t \in \bbt}$ as a stochastic process adapted to a given filtration $(\fcal_t)_{t\in \bbt}$ on some probability space $(\Omega ,\fcal,\P)$. The precise probabilistic setup is introduced in Section \ref{sec:1maturity-model}. Throughout, we assume a constant risk-free rate $r \geq 0$ used for discounting all obligations. Before any concerns about default risk, the current book value of capital (with appropriate discounting) is then
\[x_i(t) + e^{-r(T-t)}\sum_{j = 1}^n L_{ji} - e^{-r(T-t)}\bar p_i.\]
While we refer to this as the book value, one may wish to think of the external assets $x_i$ as arising from some combination of historical and market values according to their composition. We abstract this away by assuming that $x_i$ evolves according to some stochastic process.

As regards the interbank assets $L_{ji}$, depending on the probabilities of default, these will not be valued at their (discounted) face value, but instead at some intrinsic mark-to-market value that we denote by $p_{ji}(t)$ at time $t$. Figure~\ref{fig:balance_sheet} provides a visual depiction of both the book values and the realized balance sheet for an arbitrary bank $i$ in the financial system.

\begin{figure}[h!]
	\centering
	\begin{tikzpicture}
		\draw[draw=none] (0,6.5) rectangle (6.5,7) node[pos=.5]{\bf Book Value at Time $t$};
		\draw[draw=none] (0,6) rectangle (3.25,6.5) node[pos=.5]{\bf Assets};
		\draw[draw=none] (3.25,6) rectangle (6.5,6.5) node[pos=.5]{\bf Liabilities};
		
		\filldraw[fill=blue!20!white,draw=black] (0,4) rectangle (3.25,6) node[pos=.5,style={align=center}]{External \\ $x_i(t)$};
		\filldraw[fill=yellow!20!white,draw=black] (0,0) rectangle (3.25,4) node[pos=.5,style={align=center}]{Interbank \\ $e^{-r(T-t)}\sum_{j = 1}^n L_{ji}$};
		
		\filldraw[fill=purple!20!white,draw=black] (3.25,3) rectangle (6.5,6) node[pos=.5,style={align=center}]{Total \\ $e^{-r(T-t)}\bar p_i$};
		\filldraw[fill=orange!20!white,draw=black] (3.25,0) rectangle (6.5,3) node[pos=.5,style={align=center}]{Capital \\ $x_i(t) - e^{-r(T-t)}\bar p_i$ \\ $+ e^{-r(T-t)}\sum_{j = 1}^n L_{ji}$};
		
		\draw[->,line width=1mm] (7,3) -- (8,3);
		
		\draw[draw=none] (8.5,6.5) rectangle (15,7) node[pos=.5]{\bf Realized Balance Sheet at Time $t$};
		\draw[draw=none] (8.5,6) rectangle (11.75,6.5) node[pos=.5]{\bf Assets};
		\draw[draw=none] (11.75,6) rectangle (15,6.5) node[pos=.5]{\bf Liabilities};
		
		\filldraw[fill=blue!20!white,draw=none] (8.5,4) rectangle (11.75,6) node[pos=.5,style={align=center}]{External \\ $x_i(t)$};
		\filldraw[fill=yellow!20!white,draw=black] (8.5,1) rectangle (11.75,4) node[pos=.5,style={align=center}]{Interbank \\ $\sum_{j = 1}^n p_{ji}(t)$};
		\filldraw[fill=yellow!20!white,draw=black] (8.5,0) rectangle (11.75,1);
		
		\filldraw[fill=purple!20!white,draw=black] (11.75,3) rectangle (15,6) node[pos=.5,style={align=center}]{Total \\ $e^{-r(T-t)}\bar p_i$};
		\filldraw[fill=orange!20!white,draw=black] (11.75,1) rectangle (15,3) node[pos=.5,style={align=center}]{Capital \\ $x_i(t) - e^{-r(T-t)}\bar p_i$ \\ $+ \sum_{j = 1}^n p_{ji}(t)$};
		\filldraw[fill=orange!20!white,draw=black] (11.75,0) rectangle (15,1);
		\draw (8.5,0) rectangle (15,6);
		\draw (11.75,0) -- (11.75,6);
		
		\begin{scope}
			\clip (8.5,0) rectangle (15,1);
			\foreach \x in {-8.5,-8,...,15}
			{
				\draw[line width=.5mm] (8.5+\x,0) -- (15+\x,6);
			}
		\end{scope}
	\end{tikzpicture}
	\caption{Stylized book and balance sheet for a firm at time $t$ before maturity of interbank claims.}
	\label{fig:balance_sheet}
\end{figure}

Let  $P_j(T,\omega)\in \{0,1\}$, for $\omega \in \Omega$, denote the realized indicator of solvency for bank $j$ at the maturity time $T$. Within our stylized balance sheet construction, the key modelling choice is then to let the mark-to-market values of the obligations from bank $j$ to $i$ be given by $p_{ji}(t) := e^{-r(T-t)} L_{ji} (\beta + (1-\beta)\E[P_j(T) \; | \; \fcal_t]) $,
where $\beta \in [0,1]$ is the recovery rate. Defining $P_j(t) := \E[P_j(T) \; | \; \fcal_t] $ as bank $j$'s (conditional) probability of solvency at maturity, as measured at time $t$, we can write these  mark-to-market values as 
\begin{equation}\label{eq:MtM_values}
p_{ji}(t) = e^{-r(T-t)} L_{ji}(\beta+(1-\beta)P_j(t)),
\end{equation}
which we note belong (almost surely) to the intervals $e^{-r(T-t)} L_{ji}\times[\beta,1]$. In this way, the realized balance sheet for bank $i$ has (possible) write-downs in the value of assets and, correspondingly, the realized capital of bank $i$ at time $t$ is in turn given by
\begin{equation}\label{eq:realized_capital}
	K_i(t) = x_i(t) + e^{-r(T-t)} \sum_{j = 1}^n L_{ji} \bigl(\beta+(1-\beta)P_j(t) \bigr) - e^{-r(T-t)}\bar p_i.
	\end{equation}

It remains to clarify how insolvency is determined. As in the static setting of Section \ref{sec:setting-network}, we equate solvency to positive capital. In our dynamic setup, bank $i$ will thus default on their obligations at the first time their realized capital drops below zero, i.e., at the stopping time
\begin{equation}\label{eq:default}
\tau_i := \inf\bigl\{t \in \bbt \; \left| \; K_i(t) < 0 \right.\bigr\},
\end{equation}
for $K_i$ given by \eqref{eq:realized_capital}, wherein we now have $P_j(t)= \E[P_j(T) \; | \; \fcal_t] =\P(\tau_j > T \; | \; \fcal_t)$. In particular, defaults can occur before obligations are due. This may be due to regulatory constraints, as in the introductory discussion, or it can be the result of (positive net worth) safety covenants in the lending agreements (\cite{black1976valuing, leland1994corporate}). More generally, a default barrier of the form \eqref{eq:default} may serve as an abstract indicator of any form of financial distress that leads to (immediate) default (\cite{longstaff1995}).
\begin{remark}\label{rem:dynamics}
		Though we only consider a single maturity $T$ within this section, we stress that the model is inherently dynamic on account of the possibility of early defaults. These can occur due to the dynamics of the external assets ($\vx$) and the \emph{probability} of future losses ($\vec{1}-\vP$). Since we mark interbank assets based on expected payments, the probability of solvency ($\vP$) and the capital ($\vK$) of the banks evolve dynamically in time in such a way that the default times ($\vt$) \emph{cannot} solely be determined from the realization at $T$, in contrast with a one-period model. The dynamics are illustrated by the simulated sample paths in Figure~\ref{fig:path} of Section~\ref{sec:1maturity-cs-path}.
\end{remark}
\begin{remark} We note that \eqref{eq:realized_capital}-\eqref{eq:default} amounts to a fixed point problem in the pricing of $n$ coupled digital down-and-out barrier options: $P_i(t)$ is then the current value of the barrier option with maturity $T$ and a payoff of $1$ if the capital $K_i$ never crosses zero or $0$ otherwise.
\end{remark}

The equilibrium pricing scheme \eqref{eq:realized_capital}-\eqref{eq:default} may be seen as a network valuation adjustment that can be performed by a regulator or supervisor to obtain a notion of a lower bound on the health of the firms in the financial system. Under normal market conditions, the probability of contagion is vanishingly low, so firms would operate essentially without any adjustments. In particular, we are not suggesting that the system would operate according to our model in normal times. However, following a financial shock or a sudden drop in confidence, the contagious effects, captured by spiralling mark-to-market losses, could be detrimental. Our framework makes this precise, showing how losses in confidence concerning future solvency can propagate through a given financial system, and it quantifies how severe the expected outcome may be today (or at any later time conditional on a realisation of the external assets). That this yields a natural lower bound for the health of the system is established in Appendix~\ref{sec:hpa-1}.

When comparing our mark-to-market methodology to historical price accounting, as we do in Appendix~\ref{sec:hpa-1}, it should be stressed that these two regimes co-exist in bank balance sheets. The former applies to the trading book, while the latter applies to the banking book. Herein, we subject the entire interbank exposures to mark-to-market valuation (even if they are non-traded) as a hypothetical exercise to assess the full potential for informational distress contagion. One may also be interested in a more granular model that splits each exposure $L_{ij}$ so that some fraction $\theta_{ij}\in[0,1]$ is subjected to our mark-to-market mechanism while the remaining fraction $1-\theta_{ij}$ is recorded at face value until the (possible) default of the counterparty. This can be readily achieved by combining our approach with the formalism set out in Appendix~\ref{sec:hpa-1}.

In our discussion of the interbank system, we have, so far, focused exclusively on direct links between the banks given by their interbank liabilities. However, similarly to the setting of~\cite{glasserman2016contagion}, one could also interpret each $L_{ij}>0$ as a more general (indirect) link, measuring how exposed one believes bank $j$ is to worries about the future solvency of bank $i$, e.g.~due to similarities between the two banks' operations and their general circumstances or simply because bank $j$ is perceived as vulnerable to worries about other banks in general. Returning to our discussion of SVB and the 2023 U.S.~banking crisis, such a setup could model the distress contagion seen in the U.S.~financial system as well as the global spillover, including the collapse of the already vulnerable Credit Suisse. One may also want to consider $L_{ii}>0$ to model feedback loops within perceived solvency with respect to concerns about a given bank itself, as would have been the case for SVB. Naturally, as with the generally unobservable interbank positions, any such indirect links would need to be estimated or calibrated to data in some way, but this lies beyond the scope of the present work.

While our contributions are focused on distress contagion and, correspondingly, the role of mark-to-market versus historical price accounting rules, there are several critical aspects of the 2023 banking crisis that warrant further investigation. In particular, SVB's failure materialised as a depositor run, the details of which have spurred research on supervision, the role of uninsured deposits, and the multi-faceted nature of links between interest rates and bank stability (see, e.g., \cite{acharya2023}). Moreover, there is a need to re-examine traditional assumptions about the liquidity demands of depositors and the implications for liquidity creation, bank runs, and regulation (\cite{dietrich2023}). We do not address these issues within this work.

\begin{remark}\label{rem:illiquid} Note that our solvency-based default rule implicitly covers the liquidity problem as well, since bank $i$ can cover its liabilities at maturity if and only if the realized capital at $T$ is nonnegative. In the multiple maturity setting of Section~\ref{sec:Kmaturity}, liquidity is distinct from solvency and will be modelled explicitly.
If only the liquidity problem is desired, the default times can be reformulated to $\tau_i = T+\ind{K_i(T) \geq 0}$, i.e., bank $i$ is in default if and only if it has insufficient capital at maturity to cover its obligations, akin to~\cite{banerjee2022pricing}. In this way, defaults would only occur at the terminal time and replicate the system of, e.g.,~\cite{GK10}. Notably, as the set of defaults at maturity are a subset of those determined from \eqref{eq:default}, clearing with liquidity-only defaults will result in a financial system that appears healthier than when utilizing the solvency-based defaults considered herein.
\end{remark}

\begin{remark}\label{rem:nonmarketable}
Throughout this work we assume that the interbank network is fixed even as the mark-to-market value of obligations can fluctuate. As we assume all valuations are taken w.r.t.\ the risk-neutral measure $\P$, and under the assumption that banks are risk-neutral themselves, in equilibrium the banks have no (expected) gains by altering the network structure by buying or selling interbank obligations.  Moreover, in reality, these interbank markets may not be liquid; therefore, transacting to buy or sell interbank debt could be accompanied by transaction costs and price slippage which discourage any such modifications to the network.
\end{remark}

\subsection{Mathematical formalism and clearing solutions}\label{sec:1maturity-model}

We now present the precise mathematical formalism for the clearing problem derived above, and we  then deduce the existence of minimal and maximal equilibria. Moreover, we present the characterization of this problem in terms of an endogenously determined stochastic volatility effect, as discussed in the introduction. All proofs are postponed to Appendix \ref{sec:proofs}.

\subsubsection{Multinomial tree structure for the stochasticity}\label{sec:setting-tree}
With a view towards the numerical implementation of our model, we assume that the \emph{randomness} of the financial system obeys a multinomial tree structure. This applies both to the single maturity setting, considered here, and to the multiple maturity version studied in Section \ref{sec:Kmaturity}.

To be precise, throughout we shall work on a \emph{finite} filtered probability space $(\Omega,\fcal,(\fcal_{t_l})_{l = 0}^\ell,\P)$ with times $0 =: t_0 < t_1 < ... < t_\ell := T$, $\fcal_0 = \{\emptyset,\Omega\}$, and $\fcal_T = \fcal = 2^\Omega$. In this way we can equate the times $\bbt = \{t_0,t_1,...,t_\ell\}$ with the time steps of the tree $\{0,1,...,\ell\}$. 
We let $\Omega_t$ be the set of atoms of $\fcal_t$, and, for any $\omega_{t_l} \in \Omega_{t_l}$ ($l < \ell$), we denote the successor nodes by $\succ(\omega_{t_l}) := \{\omega_{t_{l+1}} \in \Omega_{t_{l+1}} \; | \; \omega_{t_{l+1}} \subseteq \omega_{t_l}\}.$ To simplify notation, let $\lcal_t := L^\infty(\Omega,\fcal_t,\P) = \bbr^{|\Omega_t|}$ denote the space of $\fcal_t$-measurable random variables.
We use the convention that for an $\fcal_t$-measurable random variable $x \in \lcal_t$, we denote by $x(\omega_t)$ the value of $x$ at node $\omega_t \in \Omega_t$, that is $x(\omega_t) := x(\omega)$ for some $\omega \in \omega_t$ chosen arbitrarily.

For a financial system with $n$ banks, a multinomial tree with $n+1$ branches at each node permits us to efficiently encode a full correlation structure on the external asset processes.  We let $(\Omega^n,\fcal^n,(\fcal_{l\dt}^n)_{l = 0}^{T/\dt},\P^n)$ denote the filtered probability space corresponding to such a multinomial tree with constant time steps $\dt > 0$. (For simplicity, we assume throughout that $T$ is divisible by the time step $\dt$.) As there are $n+1$ branches at each node within this tree, $|\Omega_t^n| = (n+1)^{t/\dt}$ at each time $t = 0,\dt,...,T$ with equal probability $\P^n(\omega_t^n) = (n+1)^{-t/\dt}$ for any $\omega_t^n \in \Omega_t^n$.
	Because of the regularity of this system, we will index the atoms at time $t$ as $\omega_{t,i}^n \in \Omega_t^n$ for $i = 1,2,...,(n+1)^{t/\dt}$.  Similarly, we can encode the successor nodes automatically as $\succ(\omega_{t,i}^n) = \{\omega_{t+\dt,j}^n \; | \; j \in (n+1)(i-1)+\{1,...,n+1\}\}$ for any time $t$ and atom $i$. 
	
	In our numerical case studies, we shall work with a suitable geometric random walk $\vx := (\vx(0),\vx(\dt),...,\vx(T)) \in \prod_{l = 0}^{T/\dt} \lcal_{l\dt}^n$ on $(\Omega^n,\fcal^n,(\fcal_{l\dt}^n)_{l = 0}^{T/\dt},\P^n)$, giving the time $t$ value of the external assets
  $\vx(t) = (x_1(t),x_2(t),...,x_n(t))^\T$. Following the construction in~\cite{he1990convergence}, we take $\sigma = (\sigma_1,...,\sigma_n) \in \bbr^{n \times n}$ to be a nondegenerate matrix encoding the desired covariance structure $C = \sigma^2$.
	The $k^\text{th}$ element $x_k$ can then be defined recursively by
	\begin{align}
		\label{eq:gbm} x_k(t+\dt,\omega_{t+\dt,(n+1)(i-1)+j}^n) &= x_k(t,\omega^n_{t,i})\exp\left((r - \frac{\sigma_{kk}^2}{2})\dt + \sigma_k^\T \tilde{\epsilon}_j\sqrt{\dt}\right) 
	\end{align}
	for some initial point $x_k(0,\Omega^n) \in \bbr^n_{++}$ and such that $\tilde{\epsilon} = (\tilde{\epsilon}_1,...,\tilde{\epsilon}_{n+1}) \in \bbr^{n \times (n+1)}$ is generated from an $(n+1)\times(n+1)$ orthogonal matrix as in~\cite{he1990convergence}. To illustrate this construction, we revisit the simple two-bank system introduced in Example~\ref{ex:running-static} and form a corresponding tree for three time periods. We will then refer back to this particular construction in later examples when demonstrating different aspects of our model.
\begin{example}\label{ex:running-tree}
Consider the $n = 2$ bank system with $\vx(0) = (1.9,1.5)^\T$ as taken in Example~\ref{ex:running-static}. 
Fix $\dt = 0.5$ with maturity $T = 1$. For simplicity, take $r = 0$.
Further, consider volatilities \[\sigma = \frac{1}{2} \left(\begin{array}{cc} \sqrt{0.275}+\sqrt{0.225} & \sqrt{0.275}-\sqrt{0.225} \\ \sqrt{0.275}-\sqrt{0.225} & \sqrt{0.275}+\sqrt{0.225} \end{array}\right) \quad \Rightarrow \quad C = \sigma^2 = \left(\begin{array}{cc} 0.25 & 0.025 \\ 0.025 & 0.25 \end{array}\right).\]
Figure~\ref{fig:running-tree} displays the resulting tree for the external assets, constructed as in \cite{he1990convergence} according to the geometric random walk \eqref{eq:gbm}.
\end{example}

\begin{figure}[h!]
	\begin{center}
		\begin{tikzpicture}
			[
			grow=right,
			level distance=3cm,
			sibling distance=1.8cm,
			edge from parent path={(\tikzparentnode.east) -- (\tikzchildnode.west)}
			]
			\node (root) {$\left(\begin{array}{c} 1.9 \\ 1.5 \end{array}\right)$}
			child {node {$\left(\begin{array}{c} 1.6069 \\ 2.2679 \end{array}\right)$}
				child [level distance=3.5cm, sibling distance=0.6cm] {node {$(1.3590,3.4288)^\T$}}
				child [level distance=3.5cm,sibling distance=0.6cm] {node {$(2.4294,1.9180)^\T$}}
				child [level distance=3.5cm,sibling distance=0.6cm] {node {$(1.0418,1.4704)^\T$}}
			}
			child {node (mid) {$\left(\begin{array}{c} 2.8726 \\ 1.2686 \end{array}\right)$}
				child [level distance=3.5cm,sibling distance=0.6cm] {node {$(2.4294,1.9180)^\T$}}
				child [level distance=3.5cm,sibling distance=0.6cm] {node (end) {$(4.3432,1.0729)^\T$}}
				child [level distance=3.5cm,sibling distance=0.6cm] {node {$(1.8625,0.8225)^\T$}}
			}
			child {node {$\left(\begin{array}{c} 1.2319 \\ 0.9725 \end{array}\right)$}
				child [level distance=3.5cm,sibling distance=0.6cm] {node {$(1.0418,1.4704)^\T$}}
				child [level distance=3.5cm,sibling distance=0.6cm] {node {$(1.8625,0.8225)^\T$}}
				child [level distance=3.5cm,sibling distance=0.6cm] {node {$(0.7987,0.6306)^\T$}}
			};
			
			\draw[->, thick] ($(root.south) + (0,-2.6)$) -- ($(root.south) + (7.5,-2.6)$) node[pos = 1.04] {$t$};
			\foreach \x/\label in {0/$0$, 3.1/$0.5$, 6.4/$1$} {
				\draw[thick] ($(root.south) + (\x,-2.5)$) -- ($(root.south) + (\x,-2.7)$);
				\node[below] at ($(root.south) + (\x,-2.7)$) {\label};
			}
			
		\end{tikzpicture}
	\end{center}\vspace{-3pt}
	\caption{Visualization of the tree for $\vx$ in Example~\ref{ex:running-tree} constructed in the manner of \eqref{eq:gbm}.}
	\label{fig:running-tree}
\end{figure}

\subsubsection{Clearing solutions: existence and non-uniqueness}\label{sect:exist_clearing_sol}

We formulate our search for clearing solutions as an equilibrium problem that is jointly on the net worths ($\vK = (K_1,...,K_n)^\T$), survival probabilities ($\vP = (P_1,...,P_n)^\T$), and default times ($\vt = (\tau_1,...,\tau_n)^\T$). By convention, and without loss of generality, we will set $\tau_i(\omega) := T+1$ if bank $i$ does not default for the realisation $\omega \in \Omega$. 
We can, thus, take the domain for our equilibrium problem to be the complete lattice
\begin{equation*}
 \left\{(\vK,\vP,\vt) \in \left(\prod_{l = 0}^\ell \lcal_{t_l}^n\right)^2 \!\! \times  \{t_0,t_1,...,t_\ell,T+1\}^{|\Omega| \times n} \; \left| \; \begin{array}{l} \forall t = t_0,t_1,...,t_\ell: \\ \vK(t) \in \vx(t) + e^{-r(T-t)}\left([\vec{0} \; , \; \vL^\T\vec{1}] - \bar\vp\right) \\ \vP(t) \in [\vec{0} \; , \; \vec{1}] \end{array} \!\!\!\right.\right\}
\end{equation*}
which we denote by $\bbd^T$. The corresponding clearing system is then mathematically constructed by means of the mapping $\Psi^T : \bbd^T\rightarrow \bbd^T$ given by
\begin{align}
\label{eq:1maturity} & \Psi^T(\vK,\vP,\vt) := (\Psi^T_{\vK}(t_l,\vP(t_l)) \; , \; \Psi^T_{\vP}(t_l,\vt) \; , \; \Psi^T_{\vt}(\vK))_{l = 0}^\ell \\[4pt]
\nonumber &\begin{cases} 
    \Psi^T_{\vK,i}(t,\tilde{\vP}) = x_i(t) + e^{-r(T-t)}\sum_{j = 1}^n L_{ji} (\beta+(1-\beta)\tilde{P}_j) - e^{-r(T-t)}\bar p_i \\
    \Psi^T_{\vP,i}(t,\vt) = \P(\tau_i > T \; | \; \fcal_t) \\
    \Psi^T_{\vt,i}(\vK) = \inf\{t \in \bbt \; | \; K_i(t) < 0\} 
    \end{cases} \;\; \forall i = 1,...,n.
\end{align}
A clearing solution $(\vK^*,\vP^*,\vt^*)$ to our equilibrium problem is defined as a fixed point of $\Psi^T$ defined by \eqref{eq:1maturity}. Existence of clearing solutions is guaranteed by the following result.

\begin{theorem}\label{thm:1maturity-exist}
The clearing solutions $\{(\vK^*,\vP^*,\vt^*) \in \bbd^T \; | \; (\vK^*,\vP^*,\vt^*) = \Psi^T(\vK^*,\vP^*,\vt^*)\}$ to ~\eqref{eq:1maturity} form a lattice in $\bbd^T$ with greatest and least solutions  $(\vK^\uparrow,\vP^\uparrow,\vt^\uparrow) \geq (\vK^\downarrow,\vP^\downarrow,\vt^\downarrow)$.
\end{theorem}

As we show next, there will generally \emph{not} be a unique clearing solution. However, the maximal clearing solution is, of course, unique by definition. As the maximal solution is predominantly utilized in the systemic risk literature (see, e.g., \cite[Remark 3.4]{RV13} and \cite{AW_15}), we will focus on it throughout. 
Though we focus on the maximal solutions, all results presented regarding that solution also hold for the minimal solution.

\begin{proposition}\label{ex:nonunique}
	The equilibrium problem \eqref{eq:1maturity} exhibits non-uniqueness of clearing solutions.
\end{proposition}

To establish this claim, we shall continue to consider the setting of our running example consisting, so far, of Examples~\ref{ex:running-static} and~\ref{ex:running-tree}. In this case, we can readily construct the maximal and minimal solutions which turn out to be distinct, as per Example \ref{ex:running-nonunique} below.
\begin{example}\label{ex:running-nonunique}
Consider the $n = 2$ bank network from Example~\ref{ex:running-static} with the obligations due at time $T = 1$ under the tree dynamics of Example~\ref{ex:running-tree}.
The maximal solution $(\vK^\uparrow,\vP^\uparrow)$ is displayed in Figure~\ref{fig:running-nonunique}. In particular, $\vP^\uparrow(0) = (5/9,1/3)^\T$ which we note is different from the maximal static solution of $\vec{1}$ from Example~\ref{ex:running-static}.
In fact, as highlighted within Remark~\ref{rem:dynamics}, we wish to note that bank 2 would be solvent in states $\omega_{1,5}^2$ and $\omega_{1,6}^2$ if we only considered clearing at $T$, i.e., without early defaults the probability of solvency would be $(5/9,5/9)^\T$.
By constrast, the minimal solution is equivalent to the minimal static solution provided in Example~\ref{ex:running-static}, i.e., with both banks being in default at time $t = 0$ so that $\vP^\downarrow = \vec{0}$ with $\vK^\downarrow(0) = (-0.1,-0.5)^\T$.
\end{example}

\begin{figure}[h]
\hspace{12pt}	\begin{subfigure}[t]{0.54\textwidth}
		\begin{center}
			\begin{tikzpicture}
				[
				grow=right,
				level distance=3cm,
				sibling distance=1.8cm,
				edge from parent path={(\tikzparentnode.east) -- (\tikzchildnode.west)}
				]
				\node (root) {$\left(\begin{array}{c} 0.2333 \\ 0.0556 \end{array}\right)$}
				child {node {$\left(\begin{array}{c} 0.6069 \\ 1.2679 \end{array}\right)$}
					child [level distance=3.5cm, sibling distance=0.6cm] {node {$(0.3590,2.4288)^\T$}}
					child [level distance=3.5cm,sibling distance=0.6cm] {node {$(1.4294,0.9180)^\T$}}
					child [level distance=3.5cm,sibling distance=0.6cm] {node {$(0.0418,0.4704)^\T$}}
				}
				child {node (mid) {$\left(\begin{array}{c} 0.8726 \\ -0.0648 \end{array}\right)$}
					child [level distance=3.5cm,sibling distance=0.6cm] {node {$(0.4294,\ast)^\T$}}
					child [level distance=3.5cm,sibling distance=0.6cm] {node (end) {$(2.3432,\ast)^\T$}}
					child [level distance=3.5cm,sibling distance=0.6cm] {node {$(-0.1375,\ast)^\T$}}
				}
				child {node {$\left(\begin{array}{c} -0.7681 \\ -1.0275 \end{array}\right)$}
					child [level distance=3.5cm,sibling distance=0.6cm] {node {$(\ast,\ast)^\T$}}
					child [level distance=3.5cm,sibling distance=0.6cm] {node {$(\ast,\ast)^\T$}}
					child [level distance=3.5cm,sibling distance=0.6cm] {node {$(\ast,\ast)^\T$}}
				};
				\draw[->, thick] ($(root.south) + (0,-2.6)$) -- ($(root.south) + (7.5,-2.6)$) node[pos = 1.04] {$t$};
				\foreach \x/\label in {0/$0$, 3.1/$0.5$, 6.4/$1$} {
					\draw[thick] ($(root.south) + (\x,-2.5)$) -- ($(root.south) + (\x,-2.7)$);
					\node[below] at ($(root.south) + (\x,-2.7)$) {\label};
				}
				
			\end{tikzpicture}
		\end{center}
		\subcaption{Clearing capitals $\vK^\uparrow$}
		\label{fig:running-nonunique-K}
	\end{subfigure}
	\begin{subfigure}[t]{0.35\textwidth}
		\begin{center}
			\begin{tikzpicture}
				[
				grow=right,
				level distance=2.3cm,
				sibling distance=1.8cm,
				edge from parent path={(\tikzparentnode.east) -- (\tikzchildnode.west)}
				]
				\node (root) {$\left(\begin{array}{c} 5/9 \\ 1/3 \end{array}\right)$}
				child {node {$\left(\begin{array}{c} 1 \\1 \end{array}\right)$}
					child [level distance=2cm, sibling distance=0.6cm] {node {$(1,1)^\T$}}
					child [level distance=2cm,sibling distance=0.6cm] {node {$(1,1)^\T$}}
					child [level distance=2cm,sibling distance=0.6cm] {node {$(1,1)^\T$}}
				}
				child {node (mid) {$\left(\begin{array}{c} 2/3 \\ 0 \end{array}\right)$}
					child [level distance=2cm,sibling distance=0.6cm] {node {$(1,0)^\T$}}
					child [level distance=2cm,sibling distance=0.6cm] {node (end) {$(1,0)^\T$}}
					child [level distance=2cm,sibling distance=0.6cm] {node {$(0,0)^\T$}}
				}
				child {node {$\left(\begin{array}{c} 0 \\ 0 \end{array}\right)$}
					child [level distance=2cm,sibling distance=0.6cm] {node {$(0,0)^\T$}}
					child [level distance=2cm,sibling distance=0.6cm] {node {$(0,0)^\T$}}
					child [level distance=2cm,sibling distance=0.6cm] {node {$(0,0)^\T$}}
				};
				
				\draw[->, thick] ($(root.south) + (0,-2.6)$) -- ($(root.south) + (4.9,-2.6)$) node[pos = 1.04] {$t$};
				\foreach \x/\label in {0/$0$, 2.3/$0.5$, 4.2/$1$} {
					\draw[thick] ($(root.south) + (\x,-2.5)$) -- ($(root.south) + (\x,-2.7)$);
					\node[below] at ($(root.south) + (\x,-2.7)$) {\label};
				}
			\end{tikzpicture}
		\end{center}
		\subcaption{Clearing solvency probabilities $\vP^\uparrow$}
		\label{fig:running-nonunique-P}
	\end{subfigure}
	\caption{Visualization of the maximal clearing solution in Example~\ref{ex:running-nonunique}}
	\label{fig:running-nonunique}
\end{figure}

\subsubsection{Stochastic volatility effects of distress contagion}\label{sect:stoch_vol}

We now turn to examine the equilibrium effects of distress contagion on the volatility of bank capital. For concreteness, we focus on external assets that follow (correlated) geometric random walks on a multinomial tree, as defined by \eqref{eq:gbm}. With the notation of Section \ref{sec:setting-tree}, let $(\Omega^n,\fcal^n,(\fcal^n_{l\dt})_{l = 0}^{T/\dt},\P^n)$ denote the filtered probability space (where $n$ is number of banks and $\dt $ is the size of each time step). To express the dynamics of a stochastic process, say $y$, we write $\Delta y(t):=y(t+\Delta t)-y(t)$ for the forward increment at each time $t\in\{0,\ldots, T-\Delta t\}$.

	\begin{proposition}\label{prop:volatility} Let $\vx$ be given by \eqref{eq:gbm} defined on $(\Omega^n,\fcal^n,(\fcal^n_{l\dt})_{l = 0}^{T/\dt},\P^n)$. For any clearing solution  $(\vK,\vP,\vt)$ from Theorem \ref{thm:1maturity-exist}, consider the discounted processes $\tilde{\vK}(t):= e^{-rt}\vK(t)$ and $\tilde{\vx}(t):=e^{-rt}\vx(t)$. Then, the clearing solution provides adapted vector processes $(\boldsymbol{\theta}_j(t))_{t=0}^{T-\Delta t}$, for $j=1\ldots n$, such that the discounted capital of each bank has dynamics of the form
		\begin{equation}\label{eq:dynamic_form_volatility}
			\Delta \tilde{K}_i(t) = \Delta \tilde{x}_i(t) +(1-\beta)\sum_{j = 1}^n L_{ji} \boldsymbol{\theta}_j(t)^\T  \Delta\tilde{\vx}(t),\quad i=1,\ldots,n,
		\end{equation}
		subject to initial conditions $\tilde{K}_i(0)= \tilde{x}_i(0)+ e^{-rT}\sum_{j = 1}^n L_{ji} (\beta+(1-\beta)\P(\tau_j > T)) - e^{-rT}\bar p_i $ and terminal conditions $\tilde{K}_i(T)= \tilde{x}_i(T)+ \sum_{j = 1}^n L_{ji} (\beta+(1-\beta)\ind{ \tau_j > T}) - \bar p_i$ for $i=1,\ldots,n$.
	\end{proposition}
	
	The above confirms that the clearing problem amounts to a problem of identifying suitable stochastic volatility processes for the capital $\vK= (K_1,...,K_n)^\T$. Specifically, \eqref{eq:dynamic_form_volatility} shows how our notion of distress contagion induces an additional (endogenous) stochastic component to the volatility of equity returns (expressed in terms of capital $\vK$), as discussed in the introduction.
	
	We can observe the following about each process $\boldsymbol{\theta}_j$ coming from the exposure of bank $i$ to the default of its $j^\text{th}$ counterparty. If the counterparty $j$ is healthy at time $t$ (i.e., $\tilde{K}_j(t)$ is significantly above zero), then the term $\boldsymbol{\theta}_j(t)^\T  \Delta\tilde{\vx}(t)$ will be negligible independently of $\Delta\tilde{\vx}(t)$, highlighting that the additional volatility component $\boldsymbol{\theta}_j(t)$ is negligible at such times. On the other hand, when the $j^\text{th}$ counterparty is distressed (i.e., $\tilde{K}_j(t)$ is not too far from zero), then even small moves in $\Delta\tilde{\vx}(t)$ can have a dramatic effect on $\boldsymbol{\theta}_j(t)^\T  \Delta\tilde{\vx}(t)$, emphasising that the volatility component $\boldsymbol{\theta}_j(t)$ will be substantial at such times and may in fact be arbitrarily large. When the increments $\Delta\tilde{\vx}(t)$ are pushing the capital of unhealthy banks in a negative direction, this effect will be more pronounced. In summary, the dynamics \eqref{eq:dynamic_form_volatility} display volatility clustering and a marked `down-market effect' in the form of a positive correlation between higher volatility and negative returns (measured in terms of capital) in a way that is more significant the more stress there is on the system. For a simple illustration, see Figure \ref{fig:path} in Section \ref{sec:1maturity-cs-path}.

	\begin{remark}For general multinomial trees, one could obtain expressions similar to \eqref{eq:dynamic_form_volatility}, only of a more abstract form, by e.g.~following the arguments of \cite[Section IV.3]{protter2005stochastic}.
	\end{remark}

	\subsection{Forward-backward representation and dynamic programming}\label{sect:fwd_bwd_dynprog}

In this section, we consider two useful reformulations of the mathematical model for the clearing problem. The first exploits Bayes' rule to provide a forward-backward equation that is tractable computationally (Section~\ref{sec:1maturity-recursion}). The second establishes a dynamic programming principle for the maximal clearing solution (Section~\ref{sec:1maturity-dpp}). Furthermore, the latter recovers Markovianity on an extended state space and forms the basis for our study of multiple maturities in Section \ref{sec:Kmaturity}.

\subsubsection{A recursive forward-backward representation}\label{sec:1maturity-recursion}

The formulation of the clearing system through $\Psi^T$ in \eqref{eq:1maturity} is seemingly complex to compute, due to the need to find the probability of solvency $\Psi^T_{\vP}$.  Within this section, we derive a backwards recursion that can be used to simplify this computation. To this end, using the notation of Section \ref{sec:setting-tree}, we define another mapping $\bar\Psi^T: \bbd^T \to \bbd^T$ by
\begin{align}
\label{eq:1maturity-recursion} \bar\Psi^T(\vK,\vP,\vt) &:= (\Psi^T_{\vK}(t_l,\vP(t_l)) \; , \; \bar\Psi^T_{\vP}(t_l,\vP(t_{[l+1]\wedge\ell}),\vt) \; , \; \Psi^T_{\vt}(\vK))_{l = 0}^\ell \\[2pt]
\bar\Psi_{\vP,i}^T(t_l,\tilde{\vP},\vt,\omega_{t_l}) &:= \begin{cases} \sum_{\omega_{t_{l+1}} \in \succ(\omega_{t_l})} \frac{\P(\omega_{t_{l+1}})\tilde{P}_i(\omega_{t_{l+1}})}{\P(\omega_{t_l})} &\text{if } l < \ell \\ \ind{\tau_i(\omega_T) > T} &\text{if } l = \ell \end{cases} &\hspace{-30pt} \forall \omega_{t_l} \in \Omega_{t_l},~i = 1,...,n.\nonumber
\end{align}
for $(\vK,\vP,\vt) \in \bbd^T$. This yields an equivalent characterization of clearing solutions, where the conditional default probabilities are computed recursively on the tree backwards in time.

\begin{proposition}\label{prop:1maturity-recursion}
The triple $(\vK,\vP,\vt) \in \bbd^T$ is a clearing solution for~\eqref{eq:1maturity} if and only if it is a fixed point of~\eqref{eq:1maturity-recursion}.
\end{proposition}

We use this result for simulations in Section \ref{sec:1maturity-cs} below. We also stress that \eqref{eq:1maturity-recursion} can be viewed as a fixed point in $\vP \in \prod_{l = 0}^\ell [0,1]^{|\Omega_{t_l}| \times n}$ only, by explicitly defining the dependence of the net worths and default times on the solvency probabilities. In this way, the joint clearing problem in $(\vK,\vP,\vt)$ becomes a forward-backward equation in the solvency probabilities $\vP$ alone, i.e.,
\begin{equation}\label{eq:1maturity-recursion-P}
\vP = \left[\bar\Psi^T_{\vP}\left(t_l,\vP(t_{[l+1]\wedge\ell}),\Psi^T_{\vt}\left(\Psi^T_{\vK}(t_k,\vP(t_k))_{k = 0}^\ell\right)\right)\right]_{l = 0}^\ell.
\end{equation}
We refer to this as a forward-backward equation since $\vP \mapsto \Psi^T_{\vt}(\Psi^T_{\vK}(\cdot,\vP))$ is calculated forward-in-time while $\bar\Psi^T_{\vP}$ is computed recursively backward-in-time. This is formalized in the following result, whose proof is an immediate consequence of the above remarks and Proposition~\ref{prop:1maturity-recursion}.
\begin{corollary}\label{cor:1maturity-recursion}
$(\vK,\vP,\vt) \in \bbd^T$ is a clearing solution of~\eqref{eq:1maturity} if and only if $\vP$ is a fixed point of~\eqref{eq:1maturity-recursion-P} with $\vK = \Psi^T_{\vK}(t_l,\vP(t_l))_{l = 0}^\ell$ and $\vt = \Psi^T_{\vt}(\Psi^T_{\vK}(t_l,\vP(t_l))_{l = 0}^\ell)$.
\end{corollary}

\subsubsection{A dynamic programming principle}\label{sec:1maturity-dpp}

We can consider the propagation of information forward-in-time by the auxiliary process $\vi $, valued in  $\{0,1\}^n$, which keeps track of whether each bank $i$ has defaulted ($\iota_i = 0$) or not ($\iota_i = 1$) up to time $t_{l-1}$ when being used as an input for the subsequent time $t_l$.

We then define $\hat\Psi^T: \bbi  \to \lcal_T \times \lcal_T$ on $\bbi := \{(t_l,\vi) \; | \; l \in \{0,1,...,\ell\}, \; \vi \in \{0,1\}^{|\Omega_{t_{[l-1]^+}}| \times n}\}$ (with $\hat\Psi^T(t,\vi) \in \hat\bbd^T(t) := \{(\vK(t),\vP(t)) \; | \; (\vK,\vP,T\vec{1}) \in \bbd^T\}$ for any time $t$) by
\begin{align}\label{eq:1maturity-dpp}
&\hat\Psi^T(t_l,\vi) := \left(\begin{array}{c}\hat\Psi^T_{\vK}(t_l,\vi) \\ \hat\Psi^T_{\vP}(t_l,\vi)\end{array}\right) \\ 
    \nonumber & \; = 
    \begin{cases} 
    \FIX_{(\tilde\vK,\tilde\vP) \in \hat\bbd^T(t_l)} \left(\!\begin{array}{c} \Psi^T_{\vK}(t_l,\tilde\vP) \\ \left[\sum_{\omega_{t_{l+1}} \in \succ(\omega_{t_l})} \frac{\P(\omega_{t_{l+1}})\hat\Psi^T_{\vP}(t_{l+1},\diag{\vi(\omega_{t_l})\ind{\tilde\vK(\omega_{t_l}) \geq \vec{0}}}}{\P(\omega_{t_l})}\right]_{\omega_{t_l} \in \Omega_{t_l}}  \end{array}\!\right) &\text{if } l < \ell \\ 
    \FIX_{(\tilde\vK,\tilde\vP) \in \hat\bbd^T(T)} \left(\begin{array}{c} \Psi^T_{\vK}(T,\tilde\vP) \\ \diag{\vi}\ind{\tilde\vK \geq \vec{0}} \end{array}\!\right) &\text{if } l = \ell
    \end{cases}
\end{align}
where $\FIX$ denotes the operator which returns the \emph{maximal} fixed point. Crucially, $\hat\Psi^T(t,\vi)$ is well-defined since the maximal fixed point exists for every time $t$ and any set of solvent banks $\vi$.
\begin{proposition}\label{prop:1maturity-dpp-define}
$\hat\Psi^T(t,\vi)$ is well-defined for every $(t,\vi) \in \bbi$, i.e., the maximal fixed point exists (and is unique) for any combination of inputs.
\end{proposition}

The above formulation provides a dynamic programming principle for the maximal clearing solution. Indeed, we have the following result.

\begin{proposition}\label{prop:1maturity-dpp}
Let $(\vK,\vP)$ be the realized solution from $\hat\Psi^T(0,\vec{1})$, i.e., $(\vK(0,\omega_0),\vP(0,\omega_0)) = \hat\Psi^T(0,\vec{1},\omega_0)$ and $(\vK(t_l,\omega_{t_l}),\vP(t_l,\omega_{t_l})) = \hat\Psi^T(t_l,\ind{\inf_{k < l} \vK(t_k,\omega_{t_k}) \geq \vec{0}},\omega_{t_l})$ for every $\omega_{t_l} \in \Omega_{t_l}$ and $l = 1,2,...,\ell$.  Then $(\vK,\vP,\Psi^T_{\vt}(\vK))$ is the \emph{maximal} fixed point to~\eqref{eq:1maturity}.
\end{proposition}

To make clear how this may be implemented in practice, we provide a pseudocode algorithm in Algorithm~\ref{alg:1maturity-dpp} that can be used to compute the maximal clearing solution at time $t = 0$. Further, we wish to  illustrate how this works in the simple context of Example~\ref{ex:running-nonunique}.
\begin{example}\label{ex:running-dpp}
Consider the same system as in Example~\ref{ex:running-nonunique} for which the maximal clearing solution was already presented in Figure~\ref{fig:running-nonunique}. In Table~\ref{tab:running-dpp}, we illustrate the procedure of applying Algorithm~\ref{alg:1maturity-dpp} to identify this solution. Within this table, `$\hspace{0.5pt}\Rightarrow$' indicates recursive function calls, `$\Leftarrow$' indicates the result of the fixed point computation at the earlier time step, and `$\Downarrow$' indicates an update to $\vi$ within the computations. The highlighted cells of this table display the maximal clearing solution $\vP^\uparrow$ at each time and state of the tree. From this solution $\vP^\uparrow$, the clearing capital $\vK^\uparrow(0) = (0.2333,0.0556)^\T$ can readily be calculated. Note that the computed clearing solution here indeed agrees with the maximal clearing solution in Figure~\ref{fig:running-nonunique}.
\end{example}

\begin{algorithm}
	\begin{algorithmic}[1]
		\Require $l \in \{0,1,...,\ell\}$
		\Require $\omega \in \Omega_{t_l}$
		\Require $\vi \in \{0,1\}^n$
		\Function{$\hat\Psi^T_{\vP}$}{$l,\vi,\omega$}
		\If{$l = \ell$}
		\State $\vP \gets \vi$
		\Repeat
		\State $\vP^0 \gets \vP$
		\State $\vP \gets \diag{\vi}\ind{\vx(T,\omega) + \vL^\T (\beta + (1-\beta) \vP^0) - \bar\vp \geq \vec{0}}$
		\Until{$\vP = \vP^0$}
		\State\Return $\vP$
		\Else 
		\Repeat
		\State $\vi^0 \gets \vi$
		\ForAll{$\bar\omega \in \succ(\omega)$}
		\State $\vP(\bar\omega) \gets \hat\Psi^T_{\vP}(l+1,\vi^0,\bar\omega)$
		\EndFor
		\State $\vP \gets \frac{1}{\P(\omega)} \sum_{\bar\omega \in \succ(\omega)} \P(\bar\omega) \vP(\bar\omega)$
		\Repeat
		\State $\vP^0 \gets \vP$
		\State $\vP \gets \diag{\vP^0}\ind{\vx(t_l,\omega) + e^{-r(T-t_l)} \vL^\T (\beta + (1-\beta) \vP^0) - e^{-r(T-t_l)} \bar\vp \geq \vec{0}}$
		\Until{$\vP = \vP^0$}
		\State $\vi \gets \diag{\vi^0}\ind{\vx(t_l,\omega) + e^{-r(T-t_l)} \vL^\T (\beta + (1-\beta) \vP^0) - e^{-r(T-t_l)} \bar\vp \geq \vec{0}}$
		\Until{$\vi = \vi^0$}
		\State\Return $\vP$
		\EndIf
		\EndFunction
		\State $\vP^\uparrow(0) \gets \hat\Psi^T_{\vP}(0,\vec{1},\Omega)$
		\State $\vK^\uparrow(0) \gets \vx(0,\Omega) + e^{-rT} \vL^\T(\beta + (1-\beta) \vP^\uparrow(0)) - e^{-rT} \bar\vp$
		\State\Return $(\vK^\uparrow(0),\vP^\uparrow(0))$
	\end{algorithmic}
	\caption{Computing the maximal clearing solution $(\vK^\uparrow(0),\vP^\uparrow(0))$ at time $t = 0$.}
	\label{alg:1maturity-dpp}
\end{algorithm}

\renewcommand{\arraystretch}{1.15}
\begin{table}
	\centering
	\resizebox{\columnwidth}{!}{
		\begin{tabular}{|lclcl|}
			\hline
			\multicolumn{1}{|c}{$\mathbf{t = 0}$} & & \multicolumn{1}{c}{$\mathbf{t = 0.5}$} & & \multicolumn{1}{c|}{$\mathbf{t = 1}$} \\ \hline \hline 
			$\vP^\uparrow(0) = \hat\Psi^T_\vP(0,\vec{1},\Omega)$ & $\Rightarrow$ & $\vP(0.5,\omega_{0.5,1}^2) = \hat\Psi^T_\vP(0.5,\vec{1},\omega_{0.5,1}^2)$ & $\Rightarrow$ & \multirow{3}{*}{$\left\{\begin{array}{l} \vP(1,\omega_{1,1}^2) = \hat\Psi^T_\vP(1,\vec{1},\omega_{1,1}^2) = \vec{0} \\ \vP(1,\omega_{1,2}^2) = \hat\Psi^T_\vP(1,\vec{1},\omega_{1,2}^2) = \vec{0} \\ \vP(1,\omega_{1,3}^2) = \hat\Psi^T_\vP(1,\vec{1},\omega_{1,3}^2) = \vec{1} \end{array} \right.$} \\
			~ & ~ & ~ & ~ & \\
			~ & ~ & $\vP(0.5,\omega_{0.5,1}^2) = \vec{0}$ & $\Leftarrow$ & \\
			~ & ~ & \multicolumn{1}{c}{$\Downarrow$} & ~ & ~ \\ \cline{5-5}
			~ & ~ & $\vP(0.5,\omega_{0.5,1}^2) = \hat\Psi^T_\vP(1,\vec{0},\omega_{0.5,1}^2)$ & $\Rightarrow$ & \multicolumn{1}{|l|}{\cellcolor{green!10}} \\ 
			~ & ~ & ~ & ~ & \multicolumn{1}{|l|}{\cellcolor{green!10}} \\ \cline{3-3}
			~ & ~ & \multicolumn{1}{|l|}{\cellcolor{green!10}$\vP^\uparrow(0.5,\omega_{0.5,1}^2) = \vec{0}$} & $\Leftarrow$ & \multicolumn{1}{|l|}{\multirow{-3}{*}{\cellcolor{green!10}$\left\{\begin{array}{l} \vP^\uparrow(1,\omega_{1,1}^2) = \hat\Psi^T_\vP(1,\vec{0},\omega_{1,1}^2) = \vec{0} \\ \vP^\uparrow(1,\omega_{1,2}^2) = \hat\Psi^T_\vP(1,\vec{0},\omega_{1,2}^2) = \vec{0} \\ \vP^\uparrow(1,\omega_{1,3}^2) = \hat\Psi^T_\vP(1,\vec{0},\omega_{1,3}^2) = \vec{0} \end{array} \right.$}} \\ \cline{3-3}\cline{5-5}
			~ & $\Rightarrow$ & $\vP(0.5,\omega_{0.5,2}^2) = \hat\Psi^T_\vP(1,\vec{1},\omega_{0.5,2}^2)$ & $\Rightarrow$ & \multirow{3}{*}{$\left\{\begin{array}{l} \vP(1,\omega_{1,4}^2) = \hat\Psi^T_\vP(1,\vec{1},\omega_{1,4}^2) = \vec{0} \\ \vP(1,\omega_{1,5}^2) = \hat\Psi^T_\vP(1,\vec{1},\omega_{1,5}^2) = \vec{1} \\ \vP(1,\omega_{1,6}^2) = \hat\Psi^T_\vP(1,\vec{1},\omega_{1,6}^2) = \vec{1} \end{array} \right.$} \\
			~ & ~ & ~ & ~ & \\
			~ & ~ & $\vP(0.5,\omega_{0.5,2}^2) = (2/3,0)^\T$ & $\Leftarrow$ & \\
			~ & ~ & \multicolumn{1}{c}{$\Downarrow$} & ~ & ~ \\ \cline{5-5}
			~ & ~ & $\vP(1,\omega_{0.5,2}^2) = \hat\Psi^T_\vP(1,(1,0)^\T,\omega_{0.5,2}^2)$ & $\Rightarrow$ & \multicolumn{1}{|l|}{\cellcolor{green!10}} \\
			~ & ~ & ~ & ~ & \multicolumn{1}{|l|}{\cellcolor{green!10}} \\ \cline{3-3} 
			~ & ~ & \multicolumn{1}{|l|}{\cellcolor{green!10}$\vP^\uparrow(0.5,\omega_{0.5,2}^2) = (2/3,0)^\T$} & $\Leftarrow$ & \multicolumn{1}{|l|}{\multirow{-3}{*}{\cellcolor{green!10}$\left\{\begin{array}{l} \vP^\uparrow(1,\omega_{1,4}^2) = \hat\Psi^T_\vP(1,(1,0)^\T,\omega_{1,4}^2) = \vec{0} \\ \vP^\uparrow(1,\omega_{1,5}^2) = \hat\Psi^T_\vP(1,(1,0)^\T,\omega_{1,5}^2) = (1,0)^\T \\ \vP^\uparrow(1,\omega_{1,6}^2) = \hat\Psi^T_\vP(1,(1,0)^\T,\omega_{1,6}^2) = (1,0)^\T \end{array} \right.$}} \\ \cline{3-3} \cline{5-5}
			~ & $\Rightarrow$ & $\vP(0.5,\omega_{0.5,3}^2) = \hat\Psi^T_\vP(0.5,\vec{1},\omega_{0.5,3}^2)$ & $\Rightarrow$ & \multicolumn{1}{|l|}{\cellcolor{green!10}} \\
			~ & ~ & ~ & ~ & \multicolumn{1}{|l|}{\cellcolor{green!10}} \\ \cline{1-1}\cline{3-3} 
			\multicolumn{1}{|l|}{\cellcolor{green!10}$\vP^\uparrow(0) = (5/9,1/3)^\T$} & $\Leftarrow$ & \multicolumn{1}{|l|}{\cellcolor{green!10}$\vP^\uparrow(0.5,\omega_{0.5,3}^2) = \vec{1}$} & $\Leftarrow$ & \multicolumn{1}{|l|}{\multirow{-3}{*}{\cellcolor{green!10}$\left\{\begin{array}{l} \vP^\uparrow(1,\omega_{1,7}^2) = \hat\Psi^T_\vP(1,\vec{1},\omega_{1,7}^2) = \vec{1} \\ \vP^\uparrow(1,\omega_{1,8}^2) = \hat\Psi^T_\vP(1,\vec{1},\omega_{1,8}^2) = \vec{1} \\ \vP^\uparrow(1,\omega_{1,9}^2) = \hat\Psi^T_\vP(1,\vec{1},\omega_{1,9}^2) = \vec{1} \end{array} \right.$}} \\ \hline
		\end{tabular}
	}
	\caption{Visualization of the application of Algorithm~\ref{alg:1maturity-dpp} in Example~\ref{ex:running-dpp}. The resulting clearing solution $\vP^\uparrow$ is highlighted by the shaded areas.}
	\label{tab:running-dpp}
\end{table}

We stress that it is this characterization that allows us to extend the model to multiple maturities in Section \ref{sec:Kmaturity}. Moreover, it confirms that, in an extended state space, the model has a Markovian structure in the following sense.

\begin{corollary}\label{lemma:markov}
Consider Markovian external assets $\vx(t_l) = f(\vx(t_{l-1}),\tilde\epsilon)$ for independent perturbations $\tilde\epsilon$ in the notation of Section \ref{sec:setting-tree}.
Let $(\vK,\vP)$ be the realized solution from $\hat\Psi^T(0,\vec{1})$ with solvency process $\vi(t_l) := \prod_{k = 0}^{l-1} \ind{\vK(t_k) \geq \vec{0}}$. Then the joint process $(\vx,\vK,\vP,\vi)$ is Markovian.
\end{corollary}

We note that it is not at all clear if one can recover Markovianity for just the capital and the indicator of defaults, i.e., the pair $(\vK,\vi)$. In fact, we suspect this will not be true in general. However, the above at least confirms that, for the maximal clearing solution, the triple $(\vK, \vP,\vi)$ is Markovian provided the external assets $\vx$ are Markovian.
\begin{remark}\label{rem:minimal}
Though not explicitly proven herein, one can derive a corresponding dynamic programming principle for the minimal solution and, in turn, the results of Corollary~\ref{lemma:markov} also hold for the minimal clearing solution.
\end{remark}

\subsection{Realized dynamics and asset correlation: two case studies}\label{sec:1maturity-cs}
Throughout this section, we let the external assets $\vx$ undergo a geometric random walk on a multinomial tree $(\Omega^n,\fcal^n,(\fcal^n_{l\dt})_{l = 0}^{T/\dt},\P^n)$, as given by the construction \eqref{eq:gbm}. We then define the discounted assets $\tilde{\vx}(t,\omega_t^n) := \exp(-rt)\vx(t,\omega_t^n)$,
for every time $t$ and state $\omega_t^n \in \Omega_t^n$.  One can check that these discounted asset values $\tilde{\vx}$ are martingales for $(\fcal^n_{l\dt})_{l = 0}^{T/\dt}$.

We shall consider two simple case studies to demonstrate key aspects of our single maturity model. First, we will present sample paths of the clearing solution to investigate the dynamic contagion mechanism in practice. Second, we will vary the exogenous part of the correlation structure between banks in order to investigate the sensitivity to this model input. For simplicity, we assume zero recovery (i.e., $\beta = 0$) throughout these case studies so that we are aligned with the Gai--Kapadia setting~(\cite{GK10}).

\subsubsection{Realized dynamics}\label{sec:1maturity-cs-path}
Consider an illustrative two bank (plus society node) system with zero risk-free rate $r = 0$.  We take the terminal time to be $T = 1$ year with time steps $\dt = 1/12$ (i.e., 1 month).  At the maturity $T$,
\[\vL_0 = \left(\begin{array}{ccc} 0.5 & 0 & 1 \\ 0.5 & 1 & 0 \end{array}\right)\]
specifies the network of obligations, where the first column gives the external liabilities. The external assets have the same variance $\sigma^2 = 0.25$ with a correlation of $\rho = 0.5$. Both banks begin with $x_i(0) = 1.5$ in external assets for $i \in \{1,2\}$. We use the forward-backward formulation~\eqref{eq:1maturity-recursion} to compute the clearing capital $\vK$ and probabilities of solvency $\vP$. Figure~\ref{fig:path} displays a single sample path of the external assets $\vx(\omega)$, capital $\vK(\omega)$, and probabilities of solvency $\vP(\omega)$.  

\begin{figure}[h]
	\centering
	\begin{subfigure}[t]{0.3\textwidth}
		\centering
		\includegraphics[width=\textwidth]{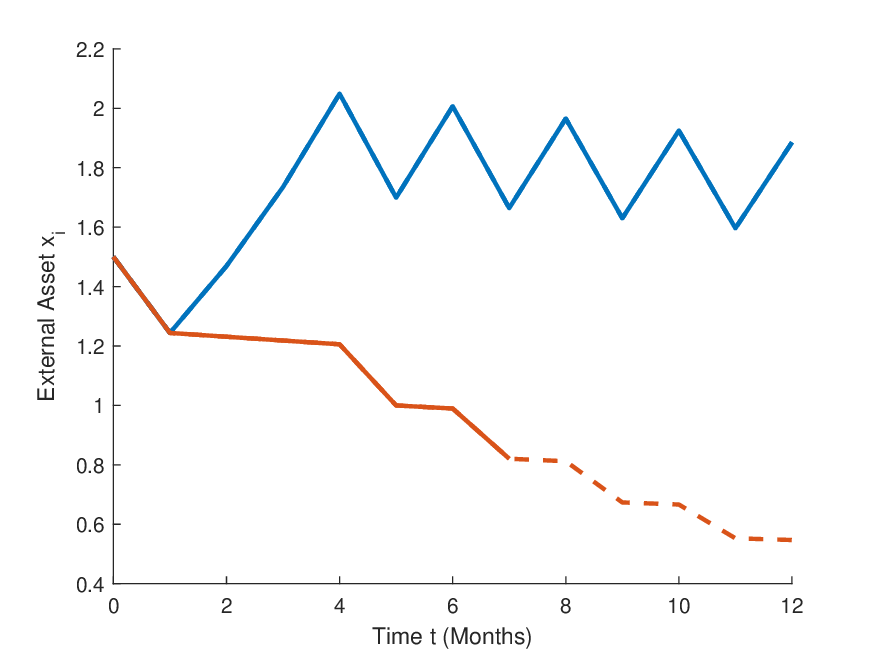}
		\caption{External assets $\vx$}
		\label{fig:path-x}
	\end{subfigure}
	~
	\begin{subfigure}[t]{0.3\textwidth}
		\centering
		\includegraphics[width=\textwidth]{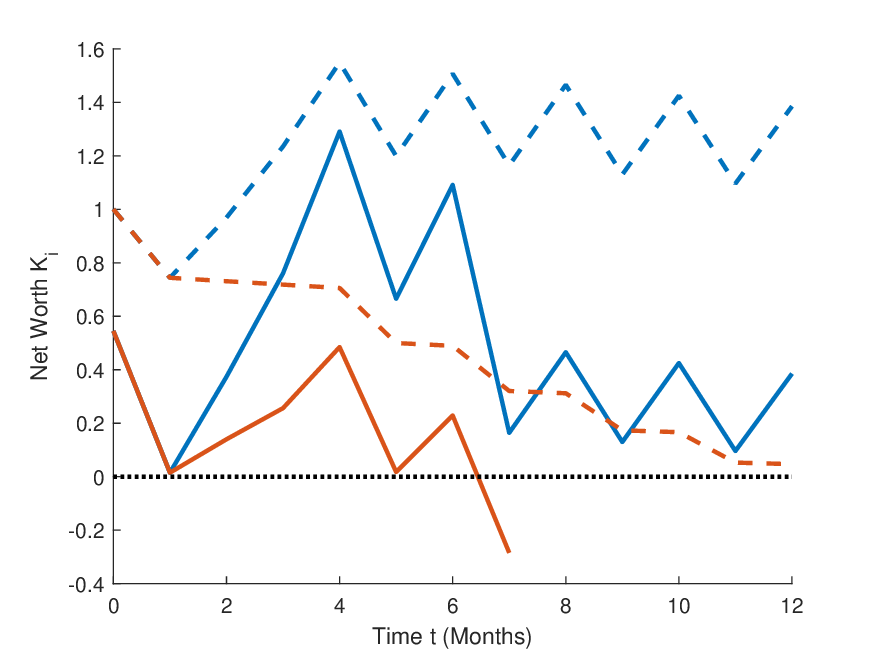}
		\caption{Capital $\vK$}
		\label{fig:path-K}
	\end{subfigure}
	~
	\begin{subfigure}[t]{0.3\textwidth}
		\centering
		\includegraphics[width=\textwidth]{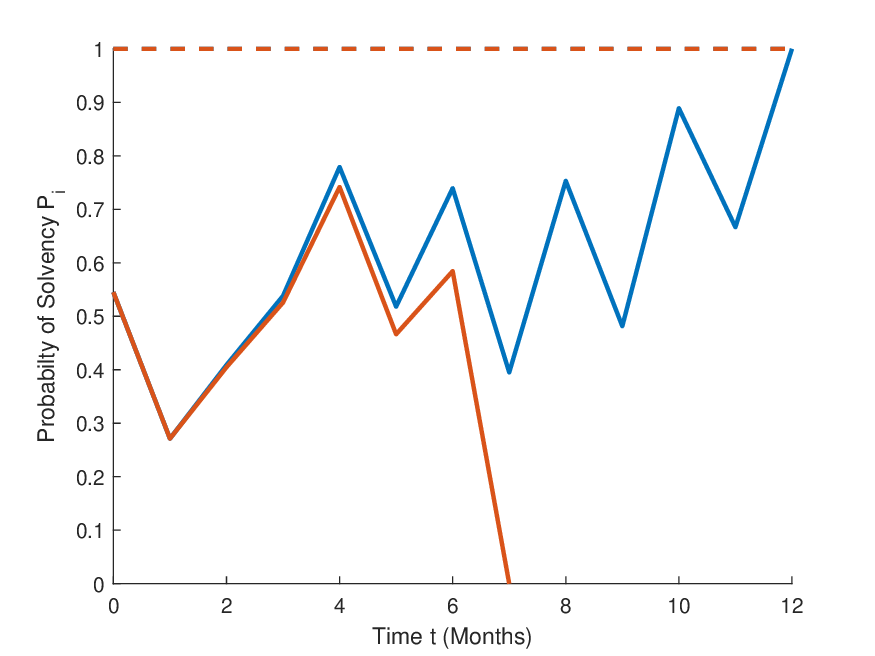}
		\caption{Probabilities of solvency $\vP$}
		\label{fig:path-P}
	\end{subfigure}
	\caption{A sample path of the external assets $\vx$, capital $\vK$, and probabilities of solvency $\vP$ for two banks. Bank 1 is displayed in blue and bank 2 in red. Our model yields the solid lines. Dashed lines indicate historical price accounting (given by the clearing problem in Appendix \ref{sec:hpa-1}).}
	\label{fig:path}
\end{figure}

Under our model's forward-backward dynamics, bank $2$ (in the solid red line) defaults on its obligations at time $t = 7$ months, while bank $1$ (in the solid blue line) remains solvent until maturity.  This is clearly seen in both Figure~\ref{fig:path-K} of capital  and Figure~\ref{fig:path-P} of the probability of solvency.  Notably, even at time $t = 7$ months, bank $2$'s external assets are approximately $x_2(7\text{ months}) \approx 0.821$; this means that, without considering distress contagion (i.e., via historical price accounting, as indicated with the dashed lines) bank $2$ would \emph{not} be in default.\footnote{Though not displayed in Figure~\ref{fig:path}, the liquidity-only default rule (as per Remark~\ref{rem:illiquid}) would realize no defaults.}
This demonstrates the dynamic impact of our endogenous distress contagion: along the given path, bank $1$ does not default, yet the possibility that it fails to repay its obligations due to uncertainty in the future drives down the capital of bank $2$ (as interbank assets are marked down) enough so that it is driven into insolvency. Thus, a doom loop type of distress contagion from bank $1$ to itself over time (via solvency contagion in unrealized paths of the tree) ultimately leads to actual default contagion. Since, in this case, a default is actualized, the capital of bank 1 remains depressed. Had bank 2 instead survived the distress, due to a better realization of its external assets from $t=6$ months onwards, then the capital of bank 1 would recover.

Finally, note that Figures~\ref{fig:path-x} and \ref{fig:path-K} illustrate how the distress contagion induces volatility clustering in the form of a down market effect. Indeed, relative to the assets or historical price accounting values, the capital of both banks displays additional downside volatility during the distress leading to bank 2's failure.

\subsubsection{Non-monotonic dependence on exogenous asset correlation}\label{sec:1maturity-cs-corr}
Consider again the simple $n = 2$ bank network of Section~\ref{sec:1maturity-cs-path}.  Rather than studying a single sample path of this system, we will instead vary the correlation $\rho \in (-1,1)$ to emphasize the nontrivial sensitivity to (exogenous) asset correlations. As we only investigate the behavior of the system at the initial time $t = 0$, and because the system is symmetric, throughout this case study we will, without loss of generality, only discuss bank $1$. In addition, within this case study, we explore the implications of three different modelling frameworks: (i) the mark-to-market system introduced herein; (ii) the corresponding model with historical price accounting rules considered in Appendix~\ref{sec:hpa-1}; and (iii) a system without interbank obligations.\footnote{We do not benchmark with the liquidity-only default rule, which only permits defaults at the terminal time, since (as stated in Remark~\ref{rem:1maturity-comparison}) such a clearing system will dominate the historical price accounting system used as a benchmark herein.}

\begin{figure}[h]
\centering
\begin{subfigure}[t]{0.45\textwidth}
\centering
\includegraphics[width=\textwidth]{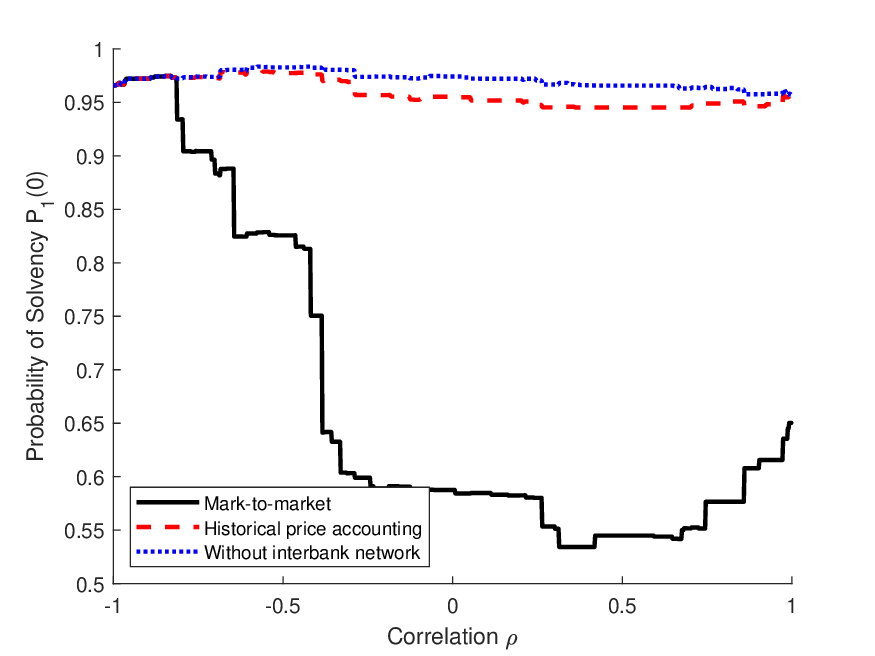}
\caption{The probability of solvency $P_1(0)$ at time $t = 0$ for bank $1$.} 
\label{fig:corr-P}
\end{subfigure}
~
\begin{subfigure}[t]{0.45\textwidth}
\centering
\includegraphics[width=\textwidth]{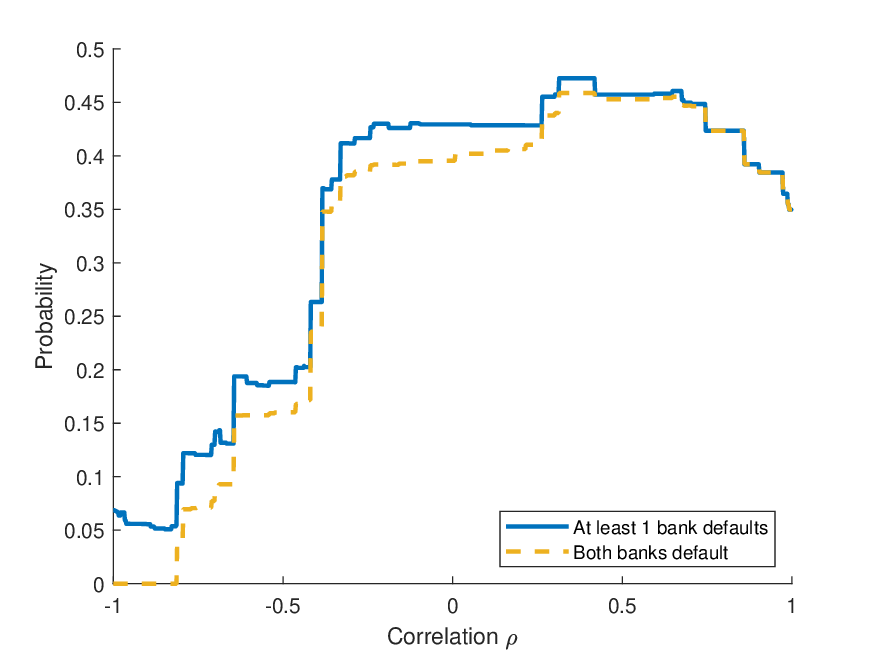}
\caption{The probability of one or both banks being in default as measured at time $t = 0$ under mark-to-market accounting.}
\label{fig:corr-defaults}
\end{subfigure}
\caption{Section~\ref{sec:1maturity-cs-corr}: The impact of the correlation $\rho \in (-1,1)$ between the external asset values on health of the financial system as measured by probabilities of solvency/default.}
\label{fig:corr}
\end{figure}

The probability of solvency $P_1(0)$ at time $t = 0$ as a function of the correlation $\rho$ between the banks' external assets is provided in Figure~\ref{fig:corr-P}; we wish to note that, due to the construction of this network, $K_1(0) = P_1(0)$.  Notably, the response of the probability of solvency to $\rho$ follows a staircase structure.  
This structure is due to the discrete nature of the tree model considered within this work.  Specifically, because of the discrete time points and asset values, for a branch of the tree to move sufficiently to cause a bank to move from solvency into default (or vice versa) requires a sufficient change in the system parameters (e.g., the correlation); once such a sufficient change occurs to the tree, due to the possibility of contagion across time and between banks, there can be knock-on effects that generate large jumps in the health of the financial system. These probabilities of solvency can be compared to the probability of solvency \emph{without} consideration for either distress contagion (i.e., using historical price accounting) or the interbank network (i.e., setting $L_{12} = L_{21} = 0$). Without an interbank network there is no avenue for contagion within this system and, as such, bank $1$ is solvent in more scenarios under this setting. In fact, the only influence that the asset correlations $\rho$ have in this no-interbank network setting is due to the construction of the multinomial tree. 
Furthermore, when distress contagion is neglected with the use of historical price accounting, the probability that bank $1$ is solvent is nearly equal to the no-interbank network setting. This implies that distress contagion can dominate, and exacerbate, the default contagion that is more commonly considered.

In Figure~\ref{fig:corr-defaults}, we investigate the probability of defaults. For this we take the view of a regulator who would be concerned with the number of defaults but not which institution is failing.  
The blue solid line displays the probability that at least one of the banks will default (as measured at time $t = 0$). As with the probability that bank $1$ is solvent $P_1(0)$, this has a step structure that is non-monotonic in the correlation $\rho$. 
The yellow dashed line displays the probability that both banks will default (as measured at time $t = 0$). Here we see that, though there is roughly a 7\% chance that there is a default in the system when $\rho \approx -1$, there is no possibility of a joint default when the banks' external assets are highly negatively correlated.  As the correlation increases, so does the likelihood of contagion in which both banks default (up until the banks default simultaneously, without the need for contagion, if $\rho = 1$). 

The above observations stand in stark contrast to the related network valuation adjustment considered in \cite{bardoscia_2019} where any correlation of the external assets has \emph{no} bearing on the individual solvency probabilities (it can be seen from \eqref{eq:bardoscia_capital} in Appendix \ref{sect:existing_works} that these only depend on the variances of the external assets, not their covariances). 
In the static setting of \cite{banerjee2022pricing}, based on the Eisenberg--Noe clearing model~(\cite{EN01,RV13}), the comonotonic scenario (i.e., $\rho = 1$) is found to yield the worst expected payments, though they do not consider the probability of default explicitly. 
With our dynamic distress contagion mechanism, we find that this no longer holds which emphasizes the complex interactions that can occur when incorporating time dynamics into financial contagion modeling. 
At first, we observe a general downward trend in the probability of solvency with the highest probability occurring near $\rho \approx -1$.  When the correlations are highly negative, for one bank to be in distress (through the downward drift of its external assets) directly means the other bank has a large surplus; in this way the interbank assets act as a diversifying investment to reduce the risk of default.
However, in contrast to the initial downward trend, we see the opposite trend from $\rho \approx 1/3$ with significant upward jumps after $\rho \approx 2/3$.
This increasing probability of solvency for large positive correlations comes from the decreased possibility for contagion to occur. Specifically, with high positive correlation, both banks are nearly guaranteed to default if one bank were to default. In contrast, for moderate correlations (e.g., $\rho \approx 1/2$), there are scenarios in which the default of one institution triggers a default in the other institution which would not have otherwise occurred. This can be seen in Figure~\ref{fig:corr-defaults} where the probability of only one bank defaulting is nearly zero for $\rho \geq 1/3$.

\section{The multiple maturity setting}\label{sec:Kmaturity}
For the remainder of the paper, we wish to allow for obligations to be due at multiple times. More precisely, at each maturity time $t_l$, there will be a given network of obligations $\vL_0^l$. Beyond this, the balance sheets are constructed according to the same principles as in Section~\ref{sec:1maturity-bs}.  

Multiple maturity models of default contagion, extending the Eisenberg--Noe framework, have recently been studied in the two-period model of \cite{KV16} and the multi-period model of \cite{BBF18}, both relying on historical price accounting. As in those works, we enforce the rule that if a bank defaults on a given obligation, then it also defaults on all subsequent obligations.  To determine a default, we note that we will now need to distinguish between solvency and liquidity, as a bank may have positive capital yet be unable to satisfy its short-term obligations in stark contrast to the single maturity setting.

In Section~\ref{sec:Kmaturity-bs}, we derive the balance sheets for each bank in the system, focusing on the new considerations that appear when multiple maturities are included. Section~\ref{sec:Kmaturity-model} provides the mathematical formalism and establishes the existence of clearing solutions.
Rather than providing detailed equivalent forward-backward and dynamic programming formulations, we only comment on how such models can be presented in this setting.
In Section~\ref{sec:Kmaturity-cs}, we use the model to study the implications of systemic risk on the term structures for interbank liabilities.
The proofs for all mathematical results are presented in Section \ref{sec:proofs-Kmaturity}.

\subsection{The balance sheet construction behind our model}\label{sec:Kmaturity-bs}

Following the balance sheet constructed within Section~\ref{sec:1maturity-bs}, but with minor modification, banks hold \emph{three} types of assets at time $t_l$: external risky assets $x_i(t_l) \in \lcal_{t_l}$, external risk-free assets, and interbank assets $\sum_{j = 1}^n L_{ji}^k$ at time $t_k$ with $k > l$ where $L_{ji}^k \geq 0$ is the total obliged from bank $j$ to $i$ at time $t_k$ ($L_{ii}^k = 0$ so as to avoid self-dealing).  Notably, as the bank may have received interbank payments at time $t_k$ with $k \leq l$ as well, the bank may have assets held in the risk-free asset.  Specifically, at time $t_l$, the bank can split its cash holdings between its external risky and risk-free assets so that the (simple) return from time $t_l$ to $t_{l+1}$ is:
\[R_i(t_{l+1},\alpha_i(t_l)) := \left[e^{r(t_{l+1}-t_l)}\alpha_i(t_l) + \frac{x_i(t_{l+1})}{x_i(t_l)}(1-\alpha_i(t_l))\right]-1\]
where $\alpha_i(t_l) \in \lcal_{t_l}$ provides the fraction of the cash account invested at time $t_l$ in the risk-free asset and, accordingly, $1-\alpha_i(t_l)$ is the fraction of the cash account invested at time $t_l$ in the risky asset. Throughout, we assume $\alpha_i(t_l,\omega_{t_l}) \in [0,1]$ so that bank $i$ is long in both assets.  The bank has liabilities $\sum_{j = 1}^n L_{ij}^k + L_{i0}^k$ at time $t_k$ where $L_{i0}^k \geq 0$ denotes the external obligations of bank $i$ at time $t_k$.  To simplify the expressions below we will define $L_{0i}^l = 0$ for all times $t_l$; furthermore to avoid the need to consider defaults due to the initial setup of the system, we will assume $L_{ij}^0 = 0$ for all pairs of banks $i,j$ so that no obligations are due at time $t_0 = 0$.

As in~\cite{KV16}, when a bank defaults on its obligations at time $t_l$, it also does so for all of its obligations at time $t_k > t_l$ as well. Moreover, when there are simultaneous defaults, there is the question of how to account for recovery of assets at the default time. Since bankruptcies take time to resolve, we assume there is a zero recovery rate at that instant, while the actual recovery rate $\beta \in [0,1]$ applies immediately after (i.e., at $t_l^+$). This is similar to the construction in~\cite{BBF18} as well as \cite[Algorithm 1]{KV16}.

Unlike the single-maturity setting, we now need to keep track of the cash account $V_i(t_l)$.  Its evolution is due to the reinvestment of the cash account from the prior time step (including any recovery of defaulting assets at time $t_{l-1}$) and the (realized) net payments (following the Gai--Kapadia setting (\cite{GK10})). This gives a recursive formulation
\begin{align*}
V_i(t_l) = (1+&R_i(t_l,\alpha_i(t_{l-1})))\biggl(V_i(t_{l-1}) + \beta\sum_{k = l-1}^{\ell} e^{-r(t_k-t_{l-1})} \sum_{j = 1}^n L_{ji}^k \ind{j \text{ defaulted at } t_{l-1}}\biggr) \\
    &+ \sum_{j = 0}^n \left[L_{ji}^l \ind{j \text{ is solvent at } t_l} - L_{ij}^l\right]
\end{align*}
for $l = 1,...,\ell$ with initial conditions $V_i(t_0) := x_i(t_0)$. For the valuation of interbank assets, we follow Section \ref{sec:1maturity-bs}.
Thus, the obligation at time $t_k$ from bank $j$ to $i$ is valued as $p_{ji}(t_l,t_k) := e^{-r(t_k-t_l)}L_{ji}^k(\beta + (1-\beta)\E[P_j(t_k,t_k) \; | \; \fcal_{t_l}])$ taking value (almost surely) in the interval $e^{-r(t_k-t_l)}L_{ji}^k\times[\beta,1]$ where $P_j(t_k,t_k,\omega_{t_k})$ is the realized indicator of solvency for bank $j$ at time $t_k$.  Defining $P_j(t,t_k) := \E[P_j(t_k,t_k) \; | \; \fcal_t]$ as the (conditional) probability of solvency for obligations with maturity at time $t_k$ as measured at time $t$, we have $p_{ji}(t_l,t_k) = e^{-r(t_k-t_l)}L_{ji}^k (\beta+(1-\beta)P_j(t_l,t_k))$.  In this way, the realized balance sheet for bank $i$ has dynamically adjusting write-downs that yield a realized net worth at time $t_l$ of
\begin{align*}
K_i(t_l) = V_i&(t_l) + \beta\sum_{j = 1}^n L_{ji}^l \ind{j \text{ defaults at } t_l}\\
    &+ \sum_{k = l+1}^\ell e^{-r(t_k-t_l)} \sum_{j = 0}^n [L_{ji}^k (\beta + (1-\beta)P_j(t_l,t_k))\ind{j \text{ did not default before } t_l} - L_{ij}^k].
\end{align*}

\begin{remark}
As in Remark~\ref{rem:nonmarketable}, we assume the interbank network is fixed. As such even though future interbank assets have a value determined by $\vP$, these assets are treated as nonmarketable and cannot be used to increase short-term liquidity as encoded in the cash account. Notably, ignoring the effects of changing the network structure, if these interbank assets are treated as both liquid and marketable then the cash account $\vV$ would be identical to $\vK$.
\end{remark}

In contrast with the single maturity setting, a default can now be due to either insolvency (if the net worth becomes negative) or illiquidity (if the bank cannot meet its obligations with its current cash account).  That is, bank $i$ will default at the stopping time $\tau_i$ given by
\begin{equation}\label{eq:default_time}
\tau_i := \inf\{t \in \bbt \; | \; \min\{K_i(t) \; , \; V_i(t)\} < 0\},
\end{equation}
so that failure occurs at the first time that either the realized capital or cash account become negative. This addresses a classical distinction between stock and flow based reasons for financial distress (\cite{wruck1990}), here corresponding to negative net worth vis-à-vis illiquidity. Absent any frictions and with perfect agreement about expected future values, an illiquid-yet-solvent bank should be able to raise the necessary funds. Thus, we are making an implicit assumption that illiquidity comes with unsalvageable (flow based) distress, in the sense that loans or selling of illiquid assets is not possible, so the bank cannot meet its obligations. This is justified by focusing on crisis situations. In practice, the illiquid bank may be too poorly managed or may face a run and/or a credit freeze, but we do not attempt to model any such events. For an interesting theoretical analysis of interbank credit freezes, albeit in a very different setting, we refer to \cite{acemoglu2021}. We also note that many classical examples of bank failures, such as Bear Sterns and Lehman Brothers during the global financial crisis, materialised as liquidity events due to their reliance on short-term debt (see, e.g., \cite{MorrisShin2016}).
	
		The distinction between insolvency and illiquidity also appears in earlier works on multi-period interbank systems such as \cite{CC15, KV16}. Similarly to this section, \cite{KV16} studies a bonafide multiple maturity model focusing on the systemic effects of short- versus long-term liabilities, but the framework is one of traditional default contagion that can be cleared in a forward-only fashion one period at a time. In \cite{CC15}, there is no real notion of multiple maturities, since it is effectively cleared as a sequence of single-maturity models, but defaults from illiquidity remain a key concern. In fact, the focus of \cite{CC15} is to analyze the clearing outcomes when incorporating different (pre-specified) strategies of a lender of last resort that, in certain (pre-specified) cases, will step in to save an illiquid-yet-solvent bank. It could be interesting to undertake a similar analysis for our model, but we do not pursue this direction here.

\begin{remark}\label{rem:illiquid2}
If one prefers to focus on solvency alone, then one simply replaces the default time \eqref{eq:default_time} with $\tau_i = \inf\{t \in \bbt \; | \; K_i(t) < 0\}$. Likewise, as in Remark~\ref{rem:illiquid}, if one only wishes to model illiquidity, then one can consider $\tau_i = \inf\{t \in \bbt \; | \; V_i(t) < 0\}$ in terms of the cash account alone. All our results can be reformulated accordingly.
\end{remark}

\begin{example}\label{ex:alpha}
We highlight here four meaningful levels of the rebalancing parameter $\alpha_i$:
\begin{itemize}
\item If $\alpha_i^0(t_l) := 0$, bank $i$ reinvests its entire cash account into the external asset $x_i$ at time $t_l$.
\item If $\alpha_i^1(t_l) := 1$, bank $i$ will move all of its investments into the risk-free bond at time $t_l$; this includes any prior investment in the external asset $x_i$. Though this is feasible in our setting, we generally consider this to be an extreme scenario. 
\item If 
    \begin{align*}
    \alpha_i^L(t_l) &:= \left[\frac{\sum_{k = 0}^l e^{r(t_l-t_k)} \sum_{j = 0}^n [L_{ji}^k\ind{\tau_j > t_k}+\beta\sum_{h = k}^\ell e^{-r(t_h-t_k)}L_{ji}^h \ind{\tau_j = t_k}-L_{ij}^k]}{V_i(t_l)+\beta\sum_{k = l}^\ell e^{-r(t_k-t_l)}\sum_{j = 1}^n L_{ji}^k \ind{\tau_j = t_l}}\right]^+\\
    &= \left[1 - \frac{x_i(t_l)}{V_i(t_l)+\beta\sum_{k = l}^\ell e^{-r(t_k-t_l)}\sum_{j = 1}^n L_{ji}^k \ind{\tau_j = t_l}}\right]^+,
    \end{align*}
    then bank $i$ will use its (risky) external asset position to cover any realized net liabilities, but will never increase its external position above its original level (i.e., if bank $i$ has net interbank assets, the external asset position will be made whole and all additional assets are invested in cash). Notably, if bank $i$ is solvent, this rebalancing parameter falls between the prior two cases, i.e., $\alpha_i^L(t_l) \in [0,1]$.
\item If
    \[\alpha_i^*(t_l,\omega_{t_l}) = \left[1 - \frac{K_i(t_l,\omega_{t_l})}{w_i \theta (V_i(t_l,\omega_{t_l})+\beta\sum_{k = l}^\ell e^{-r(t_k-t_l)}\sum_{j = 1}^n L_{ji}^k \ind{\tau_j = t_{l-1}}(\omega_{t_l}))}\right]^+,\]
    then bank $i$ is maximizing its (risky) external asset position subject to the risk-weighted capital ratio constraint with risk weight $w_i > 0$ and minimal threshold $\theta > 0$. This regulatory requirement, imposed by the Basel accords, is such that the ratio of the capital to the risk-weighted assets needs to exceed some minimal ratio.
\end{itemize}
If $\alpha_i(t_l) < 0$ then, implicitly, bank $i$ is shorting its own external assets; by a no-short selling constraint, we assume this cannot occur.  
If $\alpha_i(t_l) > 1$ then, similarly, bank $i$ is borrowing at the risk-free rate solely to purchase additional units of the risky asset; as this would produce new obligations, the study of such a scenario is beyond the scope of this work.
\end{example}

\subsection{Mathematical formalism}\label{sec:Kmaturity-model}

 We can formalize the balance sheet construction as an equilibrium problem that is jointly on the net worths ($\vK = (K_1,...,K_n)^\T$), cash accounts ($\vV = (V_1,...,V_n)^\T$), survival probabilities ($\vP = (P_1,...,P_n)^\T$), and default times ($\vt = (\tau_1,...,\tau_n)^\T$).  As in the single maturity setting, without loss of generality we set $\tau_i(\omega) := T+1$ if bank $i$ does not default on $\omega \in \Omega$.  
In contrast to the single maturity setting, we consider a different domain $\bbd := \prod_{l = 0}^\ell (\lcal_{t_l}^n \times \lcal_{t_l}^n \times [\vec{0},\vec{1}]^{|\Omega_{t_l}| \times (\ell-l+1)}) \times \{t_0,t_1,...,t_\ell,T+1\}^{|\Omega| \times n}$. This specifies that $\vK,\vV$ are adapted processes, $\vP$ is a collection of adapted processes between 0 and 1, and $\vt$ is a vector of stopping times.

Making explicit the dependence on the rebalancing parameter $\va$, a clearing solution is a fixed point $(\vK,\vV,\vP,\vt)$ of the mapping $\Psi: \bbd \to \bbd$ defined by
\begin{align}
\label{eq:Kmaturity} & \Psi(\vK,\vV,\vP,\vt;\va)  := \\
\nonumber    &\;\; (\Psi_{\vK}(t_l,\vV(t_l),\vP(t_l,\cdot),\vt) \; , \; \Psi_{\vV}(t_l,\vV(t_{[l-1]\vee 0}),\vt;\va) \; , \; \Psi_{\vP}(t_l,t_k,\vt)_{k = [l+1]\wedge\ell}^{\ell} \; , \; \Psi_{\vt}(\vK,\vV))_{l = 0}^{\ell}
\end{align}
where, for any $i = 1,...,n$,
\begin{equation*}
\begin{cases}
    \Psi_{\vK,i}(t_l,\tilde\vV,\tilde\vP,\vt) = \tilde V_i + \beta\sum_{j = 1}^n L_{ji}^l\ind{\tau_j = t_l} + \sum_{k = l+1}^{\ell} e^{-r(t_k-t_l)} \sum_{j = 0}^n \left[L_{ji}^k (\beta + (1-\beta)\tilde P_j(t_k))\ind{\tau_j \geq t_l} - L_{ij}^k\right] \\
    \Psi_{\vV,i}(t_l,\tilde\vV,\vt;\va) = \begin{cases} \begin{array}{l} [1 + R_i(t_l,\alpha_i(t_{l-1}))]\left(\tilde V_i + \beta\sum_{k = l-1}^\ell e^{-r(t_k-t_{l-1})}\sum_{j = 1}^n L_{ji}^k\ind{\tau_j = t_{l-1}}\right) \\ \quad + \sum_{j = 0}^n \left[L_{ji}^l \ind{\tau_j > t_l} - L_{ij}^l\right] \end{array} & \text{if } l > 0 \\ x_i(0) & \text{if } l = 0 \end{cases} \\
    \Psi_{\vP,i}(t_l,t_k,\vt) = \P(\tau_i > t_k \; | \; \fcal_{t_l}) \\
    \Psi_{\vt,i}(\vK,\vV) = \inf\{t \in \bbt \; | \; \min\{K_i(t),V_i(t)\} < 0\}.
\end{cases} 
\end{equation*}

In contrast to the single maturity setting, we can no longer guarantee monotonicity of the clearing problem, due to the recovery rate $\beta$ and the potential dependence of the rebalancing strategy to the capital and cash accounts (see, for instance, $\va^L$ in Example~\ref{ex:alpha}). However, it turns out that we can establish the existence of a clearing solution constructively using an extension of the dynamic programming principle formulation of the problem, following on from Section~\ref{sec:1maturity-dpp} in the single maturity setting. This construction is provided in Algorithm~\ref{alg:Kmaturity-dpp} and within Section \ref{sec:proof_multi_exist}.

\begin{theorem}\label{thm:Kmaturity-exist}
Fix the rebalancing strategies $\va(t_l,\hat\vK,\hat\vV) \in [\vec{0},\vec{1}]^{|\Omega_{t_l}|}$ so that they depend only on the current time $t_l$, capital $\hat\vK \in \lcal_{t_l}^n$, and cash account $\hat\vV \in \lcal_{t_l}^n$.
There exists a (finite) clearing solution $(\vK^*,\vV^*,\vP^*,\vt^*) = \Psi(\vK^*,\vV^*,\vP^*,\vt^*)$ to~\eqref{eq:Kmaturity}.
Furthermore, if we have Markovian external assets, $\vx(t_l) = f(\vx(t_{l-1}),\tilde\epsilon(t_l))$ for i.i.d.\ perturbations $\tilde\epsilon$, then there exists a clearing solution such that $(\vx(t_l),\vK^*(t_l),\vV^*(t_l),\vP^*(t_l,t_k)_{k = l}^\ell,\vi^*(t_l))_{l = 0}^\ell$ is Markovian where $\vi^*(t_l) := \ind{\vt^* \geq t_l}$ is the realized solvency process.
\end{theorem}

As mentioned above, we prove the results of Theorem~\ref{thm:Kmaturity-exist} by constructing the \emph{locally} maximal clearing solution. That is, following the logic of the fictitious default algorithm of~\cite{EN01}, we find the minimal set of defaulting institutions at each time $t \in \bbt$ given the future dynamics as in Section~\ref{sec:1maturity-dpp}. Analogously to Remark~\ref{rem:minimal}, we note that one can modify the provided construction to instead consider the \emph{locally} minimal solution instead.

Much like in Section~\ref{sec:1maturity-dpp}, we conclude by providing the constructive algorithm for a clearing solution $(\vK^*,\vV^*,\vP^*)$ as utilized in the proof of Theorem~\ref{thm:Kmaturity-exist} (see Section~\ref{sec:proof_multi_exist}). This is presented in Algorithm~\ref{alg:Kmaturity-dpp} below. Furthermore, we accompany this by a simple example, wherein we compute the maximal clearing solutions when adapting our running example (i.e., Examples~\ref{ex:running-nonunique} and~\ref{ex:running-dpp}) to the case of two maturities, as opposed to a single maturity.

\begin{algorithm}
\begin{algorithmic}[1]
\Require $l \in \{0,1,...,\ell\}$
\Require $\omega \in \Omega_{t_l}$
\Require $\vi \in \{0,1\}^n$
\Function{$\hat\Psi$}{$l,\hat\vK,\hat\vV,\vi,\omega$}
\If{$l = 0$}
    \State $\vV \gets \vx(0,\omega)$
\Else
    \State $\bar\vR \gets [e^{r(t_{l}-t_{l-1})}\va(t_{l-1},\hat\vK,\hat\vV) + \diag{\vx(t_{l-1},\omega)}^{-1}\diag{\vx(t_{l},\omega)}(1-\va(t_{l-1},\hat\vK,\hat\vV))] - 1$
    \State $\vV \gets (I + \diag{\bar\vR})\hat\vV + (\vL^l)^\T\vi - \vL^l\vec{1}$
\EndIf
\State $\vK \gets \vV + \beta(\vL^l)^\T\vi + \sum_{k = l+1}^\ell e^{-r(t_k-t_l)}((\vL^k)^\T\vi - L^k \vec{1})$
\For{$k \in \{l,...,\ell\}$}
    \State $\vP(k) \gets \vi$
\EndFor
\Repeat
    \State $(\vK^0,\vV^0,\vP^0) = (\vK,\vV,\vP)$
    \State $\vK \gets \vV^0 + \beta(\vL^l)^\T\ind{\vK^0\wedge\vV^0 < \vec{0}} + \sum_{k = l+1}^\ell e^{-r(t_k-t_l)}[(\vL^k)^\T\diag{\vi}(\beta + (1-\beta)\vP^0(k)) - L^k \vec{1}]$
    \If{$l = 0$}
        \State $\vV \gets \vx(0,\omega)$
    \Else
        \State $\vV \gets (I + \diag{\bar\vR})\hat\vV + (\vL^l)^\T\diag{\vi}\ind{\vK^0\wedge\vV^0 \geq \vec{0}} - L^l \vec{1}$
    \EndIf
    \State $\vP(l) \gets \diag{\vi}\ind{\vK^0\wedge\vV^0 \geq \vec{0}}$
    \If{$l < \ell$}
        \ForAll{$\bar\Omega \in \succ(\omega)$}
            \State $\hat\vV^0 \gets \vV^0+\beta\sum_{k = l}^\ell (\vL^k)^\T\diag{\vi}\ind{\vK^0\wedge\vV^0 < \vec{0}}$
            \State $(\bar\vK(\bar\omega),\bar\vV(\bar\omega),\bar\vP(\bar\omega)) \gets \hat\Psi(l+1,\vK^0,\hat\vV^0,\diag{\vi}\ind{\vK^0\wedge\vV^0\geq\vec{0}},\bar\omega)$
        \EndFor
        \For{$k \in \{l+1,...,\ell\}$}
            \State $\vP(k) \gets \frac{1}{\P(\omega)} \sum_{\bar\omega \in \succ(\omega)} \P(\bar\omega) \bar\vP(k,\bar\omega)$
        \EndFor
    \EndIf
\Until{$\vK = \vK^0$ and $\vV = \vV^0$ and $\vP = \vP^0$}
\State\Return $(\vK,\vV,\vP)$
\EndFunction
\State\Return $(\vK^*(0),\vV^*(0),\vP^*(0)) \gets \hat\Psi(0,\vx(0,\Omega),\vx(0,\Omega),\vec{1},\Omega)$
\end{algorithmic}
\caption{Computing the clearing solution $(\vK^*(0),\vV^*(0),\vP^*(0,\cdot))$ at time $t = 0$.}
\label{alg:Kmaturity-dpp}
\end{algorithm}

\begin{example}\label{ex:running-Kmaturity}
Consider the same system as in Example~\ref{ex:running-nonunique} but with the obligations split between times $t = 0.5$ and $t = 1$.
Specifically, we vary the obligations by a parameter $\lambda \in [0,1]$ so that
\[\vL_0^1 = \left(\begin{array}{ccc} 0 & 0 & \lambda \\ \lambda & 0 & \lambda \end{array}\right) \quad \text{ and } \vL_0^2 = \left(\begin{array}{ccc} 0 & 1 & 1-\lambda \\ 1-\lambda & 0 & 1-\lambda \end{array}\right).\]
In Figure~\ref{fig:running-Kmaturity}, we then plot the clearing results as observed at time $t = 0$ (we exclude the cash account $\vV^*(0) = \vx(0)$ as it is a constant by construction). We note that the $\lambda = 0$ case exactly coincides with the setting of Examples~\ref{ex:running-nonunique} and~\ref{ex:running-dpp}; we see this both for the clearing capital $\vK^*(0)$ and the solvency probability $\vP^*(0,1)$. The non-monotonicity exhibited by the clearing solutions, as we vary $\lambda\in [0,1]$, highlights the complex interactions that can occur when including multiple maturities.
\begin{figure}
\centering
\begin{subfigure}[t]{0.3\textwidth}
\centering
\includegraphics[width=\textwidth]{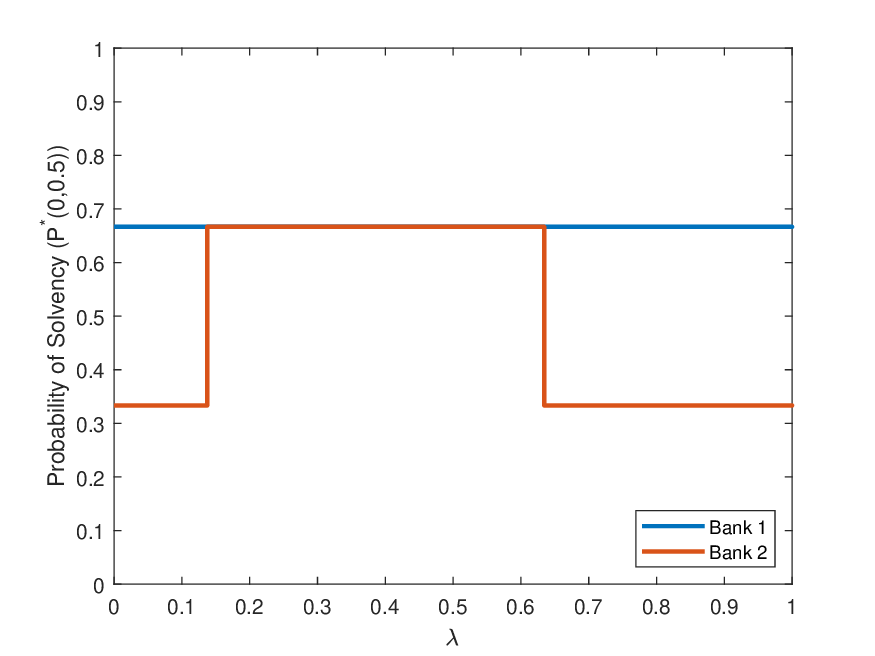}
\caption{$\vP^*(0,0.5)$}
\label{fig:running-Kmaturity-prob_l=2}
\end{subfigure}
~
\begin{subfigure}[t]{0.3\textwidth}
\centering
\includegraphics[width=\textwidth]{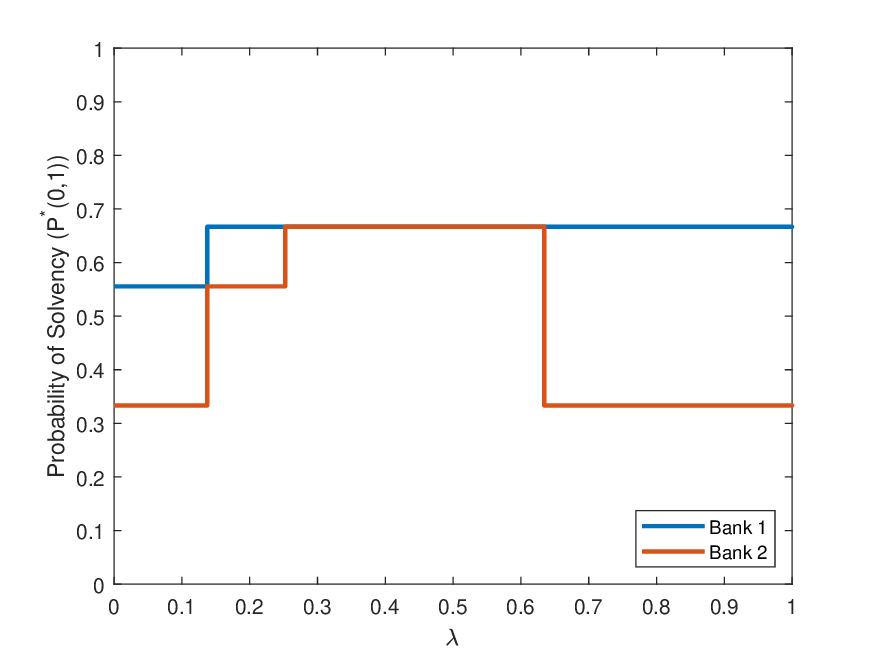}
\caption{$\vP^*(0,1)$}
\label{fig:running-Kmaturity-prob_l=3}
\end{subfigure}
~
\begin{subfigure}[t]{0.3\textwidth}
\centering
\includegraphics[width=\textwidth]{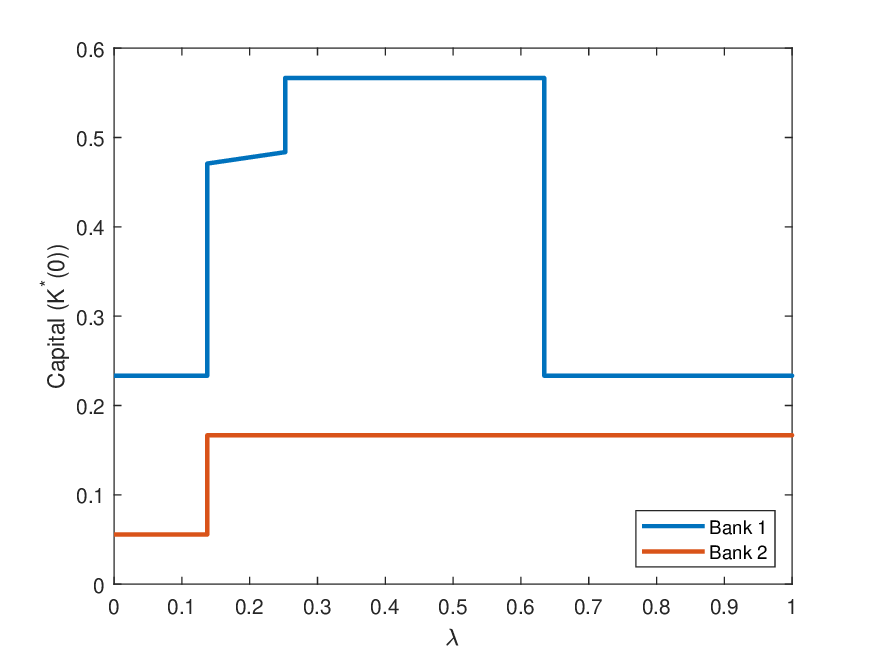}
\caption{$\vK^*(0)$}
\label{fig:running-Kmaturity-capital}
\end{subfigure}
\caption{Visualization of the clearing solution in Example~\ref{ex:running-Kmaturity} as obligations are varied between maturity dates.}
\label{fig:running-Kmaturity}
\end{figure}
\end{example}

\subsection{Systemic interbank term structures: three case studies}\label{sec:Kmaturity-cs}

For the purposes of this section, we shall consider external assets that follow a geometric random walk given by \eqref{eq:gbm}, as in Section~\ref{sec:1maturity-cs}. Furthermore, we specialize to a zero recovery rate, i.e., $\beta = 0$, as in the classical Gai--Kapadia setting (\cite{GK10}).

We will present three case studies to demonstrate the impact that our notion of distress contagion has on the (risk-aware) term structure for interbank claims.  For each bank $i$, we define this term structure, or yield curve, at time $t = 0$ through the interest rates $R_i^*(t_l) = P_i(0,t_l)^{-1/t_l}-1$ so that $(1+R_i^*(t_l))^{t_l} = 1/P_i(0,t_l)$.
First, we present a sample yield curve, as measured at time $0$, under different investment strategies $\va$.
Second, we vary the leverage of the banks to analyze the sensitivity of the term structure to the initial balance sheet of the banks. Third, we investigate a simple core-periphery network, leading to a striking observation concerning the effect that a local increase in volatility can have on the term structures of the full system. In all three case studies, we will encounter inverted shapes of the yield curves.

As in Section~\ref{sec:1maturity-cs-path}, we will compare the forward-backward approach presented herein to the forward-only pricing of historical price accounting. The latter system is defined in Appendix~\ref{sec:hpa-K} where also the corresponding clearing solutions are examined. With historical price accounting, at time $t=0$ the interbank claims for any maturity are either valued in full or at zero, depending on whether or not a given bank is in default at time $t=0$. Thus, there is not an intrinsic notion of a systemic term structure. However, in order to understand the extent to which our term structures are driven by (forward-backward) distress contagion, as opposed to (forward-only) default contagion, we can compute empirically the probabilities of default between time $t=0$ and a given maturity when running the system forward with historical price accounting. Insisting that these empirical probabilities are used to value the various interbank claims already at time $t=0$, we can then infer rates that may be compared with our term structures.

\begin{remark}
All clearing solutions computed in this section are found via the Picard iteration of~\eqref{eq:Kmaturity} beginning from the assumption that no banks will ever default.  That is, these computations are accomplished without application of the constructive dynamic programming approach presented in the proof of Theorem~\ref{thm:Kmaturity-exist}.  As this process converged for all examples, we suspect that the monotonicity of~\eqref{eq:Kmaturity} may be stronger than could be proven herein (which provides, e.g., a guarantee for the existence of a maximal clearing solution).
\end{remark}

\subsubsection{Term structures for varying investment strategies}\label{sec:Kmaturity-cs-term}
We shall work with the same $n = 2$ network as in Section~\ref{sec:1maturity-cs-path} but where all obligations $\vL_0$ are now split randomly (uniformly) over $(t_l)_{l = 1}^\ell$.  Due to the random split of obligations over time, the banks in this example are no longer symmetric institutions.  We consider only a single split of the obligations as the purpose of this case study is to understand the possible shapes of the term structure under varying rebalancing strategies $\va$.  In particular, we will study three meaningful rebalancing structures: all investments are made in the external asset only ($\va^0$ as provided in Example~\ref{ex:alpha}), all surplus interbank payments are held in the risk-free asset ($\va^L$ as defined in Example~\ref{ex:alpha}), and the optimal investment strategy ($\va^*$ as given in Example~\ref{ex:alpha} with risk-weight $w_i = 2$ for $i \in \{1,2\}$ and threshold $\theta = 0.08$ to match the Basel II Accords). 
\begin{figure}[h]
\centering
\begin{subfigure}[t]{0.45\textwidth}
\centering
\includegraphics[width=\textwidth]{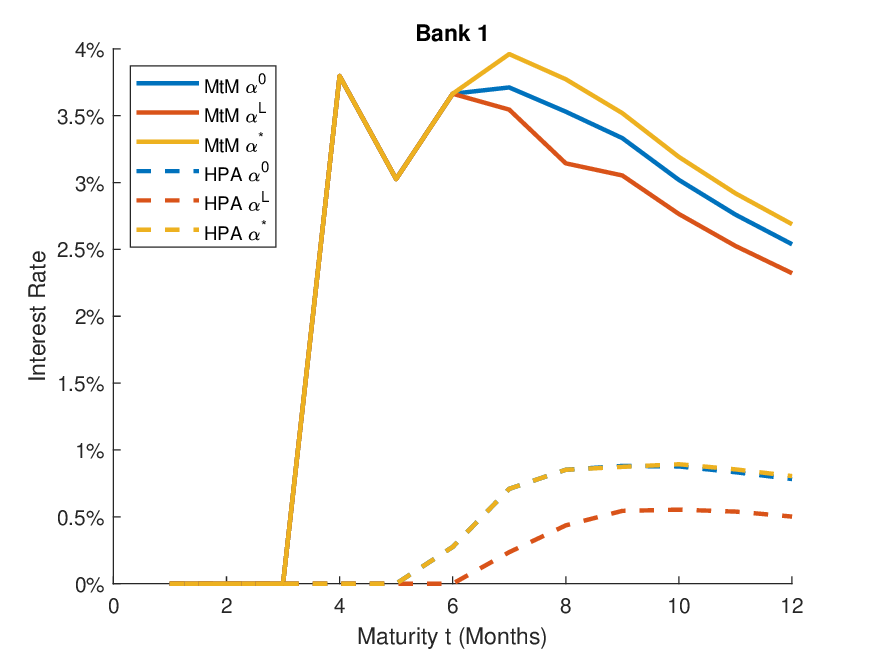}
\caption{Bank $1$}
\label{fig:term-1}
\end{subfigure}
~
\begin{subfigure}[t]{0.45\textwidth}
\centering
\includegraphics[width=\textwidth]{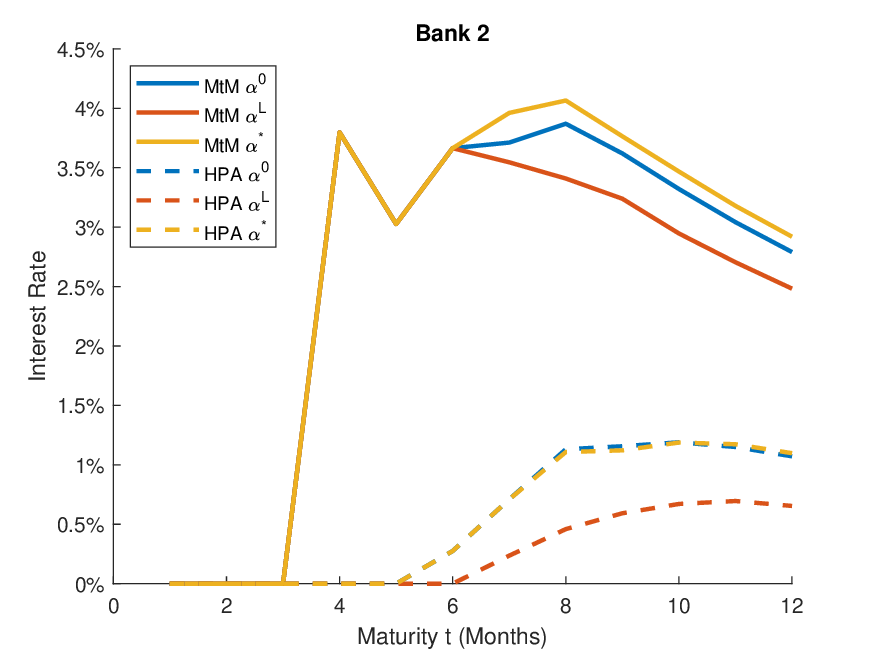}
\caption{Bank $2$}
\label{fig:term-2}
\end{subfigure}
\caption{Section~\ref{sec:Kmaturity-cs-term}: Term structure for banks $1$ and $2$ for varying investment strategies $\alpha$. The solid lines are computed from our model. The dotted lines are provided for comparison and indicate the corresponding rates under historical price accounting.}
\label{fig:term}
\end{figure}

Figure~\ref{fig:term} displays the systemic term structure for both banks in this system.  Though these interest rates are similar for the two institutions, they are not identical due to the aforementioned random splitting of the obligations over the 12 time periods (e.g., months). First, we note that the historical price accounting rules 
lead to lower interest charged than under mark-to-market accounting.
Observe that the yield curve for both banks under mark-to-market pricing has an inverted shape, i.e., the interest rate for longer dated maturities is lower than some short-term obligations.  Inverted yield curves are typically seen as precursors of economic distress; here the probabilities of default are largely driven by contagion as, e.g., bank $1$ will never default so long as bank $2$ makes all of its payments in full (bank $2$ will default in fewer than 1.1\% of paths under all three investment strategies if bank $1$ pays in full). In contrast, when considering the forward-only measure of historical price accounting without distress contagion, the inverted shape is much less pronounced.

Comparing the different investment strategies makes clear that $\va^L$ is the least risky, i.e., the interest rates are dominated by those of $\va^0,\va^*$.  This is expected; so long as the bank is healthy, it holds all surplus assets in the risk-free asset, otherwise it draws down its cash account to pay off obligations.
Surprisingly, however, $\va^*$ results in higher equilibrium interest rates than $\va^0$ under mark-to-market accounting. That is, imposing a regulatory constraint on investments through $\va^*$ leads to greater probabilities of default than solely investing in the risky (external) asset $\va^0$. Though, naively, it may seem counterintuitive that the most volatile investing strategy ($\va^0$) is \emph{not} the riskiest ex post, we conjecture this is due to the pro-cyclicality of the capital adequacy requirement (see also, e.g.,~\cite{banerjee2021price}). 
Specifically, when ignoring the possibility of counterparty risks (i.e., assuming all interbank assets are fulfilled in full) bank $2$ defaults under more scenarios when both banks follow $\va^0$ than when they follow $\va^*$.
When accounting for the network effects, the pro-cyclicality of the capital adequacy regulation implies that banks are forced to move their investments into the risk-free asset when under stress. This means that once a bank is stressed, the cash account has lower volatility and the institution is less able to recover when the external asset value increases; because of default contagion, once a bank defaults it will drag its counterparties down as well potentially precipitating a cycle of contagion.

We wish to conclude this case study by commenting briefly on the dependence of $\va^*$ on the risk-weights $w := w_1 = w_2$.  
If this risk-weight is set too low (below approximately $0.58$ for this example), then---even under the mark-to-market pricing---the banks will not be constrained at all by the regulatory environment. Therefore, under such a setting, the resulting interest rates are identical to those under $\va^0$.
Conversely, if this risk-weight is set too high (above approximately $2.15$ for this example), then the banks will not be able to invest in the risky (external) asset at all due to the regulatory constraints under mark-to-market pricing. Thus, under such a setting, the resulting interest rates are 0 (identical to those under $\va^1$ which are not displayed above) due to the construction of this system.
Around these threshold interest rates, the term structures can be highly sensitive to the regulatory environment.  Hence, naively, and heuristically, setting regulatory constraints can result in large unintended risks.  We note that the recent works of~\cite{feinstein2020capital,banerjee2021price} provide discussions on determining risk-weights to be consistent with systemic risk models.

\subsubsection{Dependence on leverage}\label{sec:Kmaturity-cs-leverage}

Having explored the impact of the investment strategy on the yield curve in Section~\ref{sec:Kmaturity-cs-term}, we now wish to explore how the (initial) leverage of the banking book can impact the shape.  As we demonstrate by numerical experiments, we find a normal term structure when the leverage is low enough but it becomes inverted for riskier scenarios.  

As in Section~\ref{sec:Kmaturity-cs-term}, we consider a variation of the $n = 2$ network of Section~\ref{sec:1maturity-cs-path} for $\ell = 12$ months, where we vary only the interbank assets and liabilities.
The banking book leverage ratio (assets over equity assuming all debts are paid in full) of bank 1 at time $t = 0$ is
\[\lambda_1 := \frac{x_1(0) + \sum_{l = 0}^\ell L_{21}^l}{x_1(0) - \sum_{l = 0}^\ell L_{10}^l} = 1.5 + \bar L_{21}\]
for $\bar L_{21} = \sum_{l = 0}^\ell L_{21}^l \geq 0$ (where $x_1(0) = 1.5$ and $\sum_{l = 0}^{\ell} L_{10}^\ell = 0.5$ by construction); similarly $\lambda_2 := 1.5 + \bar L_{12}$ for $\bar L_{12} = \sum_{l = 0}^\ell L_{12}^l \geq 0$.
As in our prior case studies, we will assume the banking book for the two banks are symmetric so that $\bar L := \bar L_{12} = \bar L_{21}$ and therefore also $\lambda := \lambda_1 = \lambda_2$ throughout this example.
In contrast to Section~\ref{sec:Kmaturity-cs-term}, here we assume all obligations are split deterministically such that
\[L_{12}^l = L_{21}^l = \begin{cases} \bar L/3 &\text{if } l \in \{3,6,12\} \\ 0 &\text{else} \end{cases} 
\quad \text{ and } \quad 
L_{10}^l = L_{20}^l = \begin{cases} 1/6 &\text{if } l \in \{3,6,12\} \\ 0 &\text{else.} \end{cases}\] 
That is, only 3 times ($t =0.25,0.5,1$) are maturities for debts, and all liabilities are split equally over those dates.
To complete the setup, we assume that all banks follow the optimal rebalancing strategy $\va^*$ (as proposed in Example~\ref{ex:alpha}) with $w := w_1 = w_2 = 2$ and $\theta = 0.08$.

\begin{figure}
\centering
\includegraphics[width=0.6\textwidth]{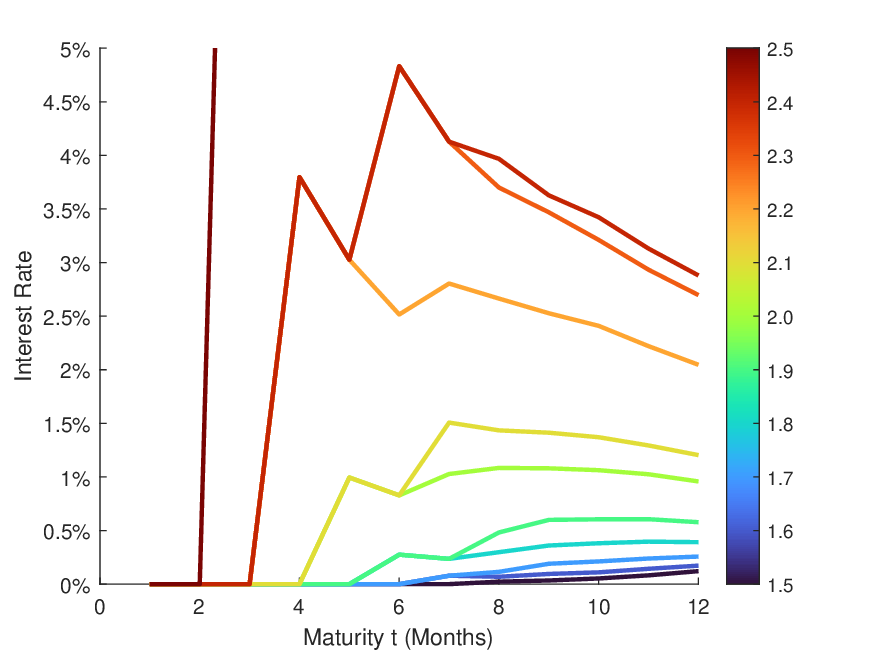}
\caption{Section~\ref{sec:Kmaturity-cs-leverage}: The term structure for both banks 1 and 2 under varying leverage ratios $\lambda \in [1.5,2.5]$ indicated by the right-hand axis.}
\label{fig:leverage}
\end{figure}
First, we want to comment on the impact that increasing the leverage $\lambda$ of the two banks has on the health of the financial system.  As seen in Figure~\ref{fig:leverage}, the system has higher (implied) interest rates $R_i^*(t_l)$ at every time $t_l$ under higher leverage.  That is, the probability of a default (as measured at time $0$) increases as the initial leverage $\lambda$ increases.  This is as anticipated because larger leverages correspond with firms that are less robust to financial stresses, i.e., a smaller shock is required to cause a bank to default.

\begin{table}[t]
\centering
{\footnotesize
\begin{tabular}{|c||c|c|c|c|c|c|c|c|c|c|c|}
\cline{2-12}
\multicolumn{1}{c|}{} & \multicolumn{11}{|c|}{Leverage $\lambda$} \\ \cline{2-12}
\multicolumn{1}{c|}{} & 1.5 & 1.6 & 1.7 & 1.8 & 1.9 & 2.0 & 2.1 & 2.2 & 2.3 & 2.4 & 2.5 \\
\cline{2-12}\hline
$t = 0.25$ & 0.00\% & 0.00\% & 0.00\% & 0.00\% & 0.00\% & 0.00\% & 0.00\% & 0.00\% & 0.00\% & 0.00\% & 16.30\% \\ \hline
$t = 0.5$ & 0.00\% & 0.00\% & 0.00\% & 0.27\% & 0.27\% & 0.83\% & 0.83\% & 2.52\% & 4.83\% & 4.83\% & 9.71\% \\ \hline
$t = 1.0$ & 0.12\% & 0.17\% & 0.26\% & 0.39\% & 0.58\% & 0.96\% & 1.20\% & 2.05\% & 2.70\% & 2.88\% & 5.18\% \\ \hline
\end{tabular}
}
\caption{Section~\ref{sec:Kmaturity-cs-leverage}: Yields for obligations as measured at time $0$.\protect\footnotemark}
\label{table:leverage}
\end{table}
\footnotetext{The system without distress contagion, as modeled by historical price accounting (Section~\ref{sec:hpa-K}) exhibits 0\% interest rates for $t = 0.25$ and $t = 0.5$; at $t = 1$, the interest rates range from $0.12\%$ to $0.20\%$ as the leverage increases.}

As obligations are only due at $t \in \{0.25,0.5,1\}$, we wish to consider the interest rates $R_i^*(t_l)$ for those dates specifically.  As displayed in Table~\ref{table:leverage}, when the leverage $\lambda$ is small ($\lambda \leq 2.1$), the interest rates charged are monotonically increasing over time, i.e., a normal yield curve.  In particular, this occurs at $\bar L = 0$ (i.e., $\lambda = 1.5$) when no interbank network exists; such a scenario can be compared with, e.g.,~\cite{black1976valuing} where a single firm is studied in isolation.  As the leverage $\lambda$ grows via the increased size of the interbank network, the risk of defaults grows as well.  This increased network size eventually leads to an inverted yield curve, i.e., in which the implied interest charged on obligations due at $t = 1$ is lower than on those due at $t = 0.5$.  Until the leverage is sufficiently high ($\lambda \approx 2.5$), no defaults are realized at the first maturity $t = 0.25$ at all.  The banks always make all payments due at $t = 0.25$ because of the tree structure considered herein; specifically, a bank fails to make payments on an early obligation only if it is already close to default at $t = 0$ as the tree does not model extreme events.

\subsubsection{Risk contamination from core to periphery}\label{sec:Kmaturity-cs-cp}
In our final case study, we consider a stylized core-periphery network, noting that such a structure has been found in several empirical studies (e.g.,~\cite{CP14,FL15,veld2014core}).  We will assume that there are $n = 12$ banks with 2 core banks, 10 peripheral banks, and a societal node.  Each core bank owes the other core bank \$3, all of the peripheral banks \$0.50, and society \$5.  The peripheral banks owe both core banks \$0.50, nothing to the other peripheral banks, and \$1 to society.
These obligations will be equally split over 4 quarters ($t = 0.25,~0.5,~0.75,~1$ with $\ell = 4$).
As before, we will assume the prevailing risk-free interest rate $r = 0$.
As in Section~\ref{sec:Kmaturity-cs-leverage}, we will assume that all banks follow the optimal investment strategy $\va^*$ presented in Example~\ref{ex:alpha} with $w_i = 2$ for every bank $i$ and with regulatory threshold $\theta = 0.08$.
Finally, the external assets are as follows.  The core banks begin at time $t = 0$ with \$15 in external assets each. The peripheral banks begin with \$3 in external assets each.  The correlation between any pair of bank assets is fixed at $\rho = 0.3$. 
We will study two scenarios for the external asset volatilities: 
\begin{enumerate}
\item a low volatility (unstressed) setting in which the volatility of the external assets for any bank is $\sigma_C^2 = \sigma_P^2 = 0.5$; 
\item a high volatility (stressed) setting in which the volatility of the external assets for either core bank jumps upward to $\sigma_C^2 = 0.75$ while the volatility for the peripheral banks is unaffected ($\sigma_P^2 = 0.5$).
\end{enumerate}
Similarly to Section~\ref{sec:1maturity-cs-corr}, we compare our mark-to-market framework with (i) the historical price accounting setting presented in Section~\ref{sec:hpa-K}, and (ii) the system without the interbank obligations. In this way we can explore the different impacts of distress and default contagion on the term structure of the banks.\footnote{As discussed above, we empirically find the probability of default under the historical price accounting and no-interbank network systems at different maturities to determine the appropriate interest rates.}

\begin{figure}[h]
\centering
\begin{subfigure}[t]{0.45\textwidth}
\centering
\includegraphics[width=\textwidth]{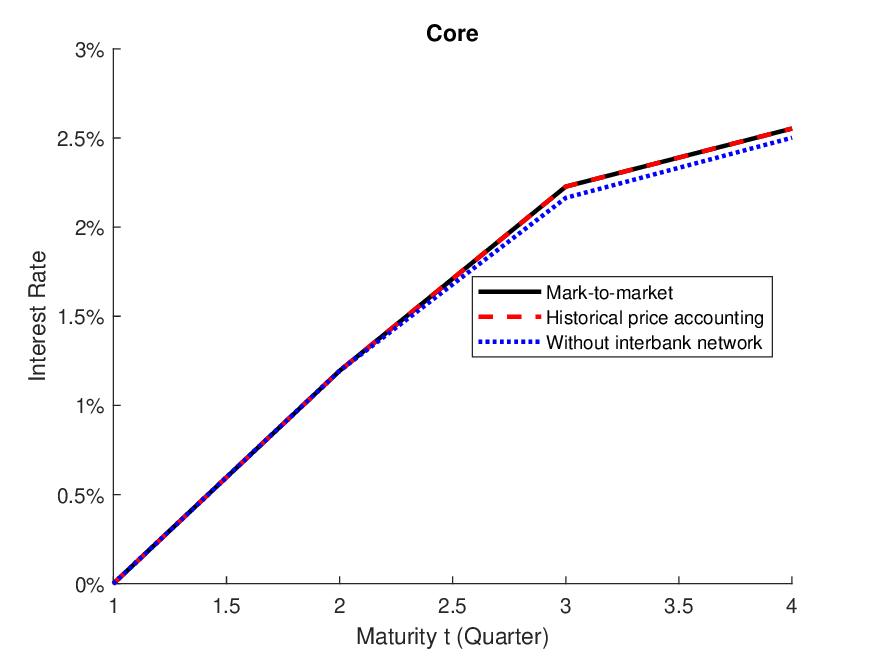}
\caption{Core: Low volatility setting}
\label{fig:cp-core-unstressed}
\end{subfigure}
~
\begin{subfigure}[t]{0.45\textwidth}
\centering
\includegraphics[width=\textwidth]{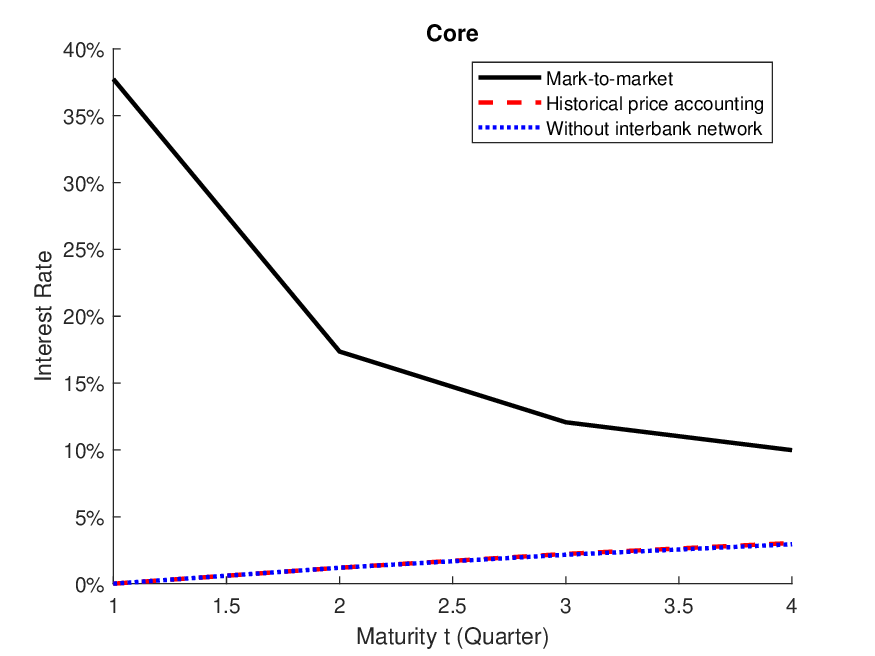}
\caption{Core: High volatility setting}
\label{fig:cp-core-stressed}
\end{subfigure}
~
\begin{subfigure}[t]{0.45\textwidth}
\centering
\includegraphics[width=\textwidth]{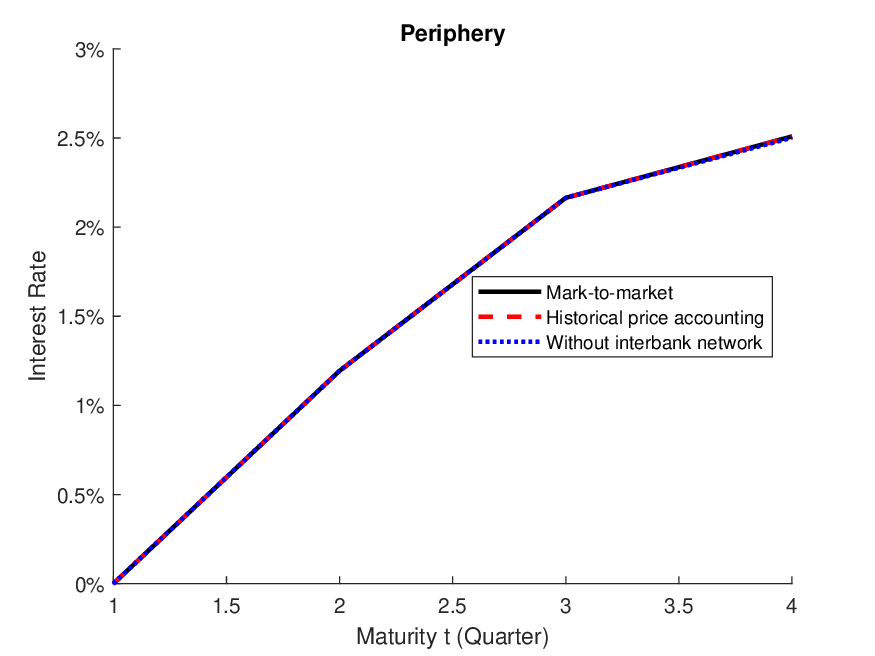}
\caption{Periphery: Low volatility setting}
\label{fig:cp-periphery-unstressed}
\end{subfigure}
~
\begin{subfigure}[t]{0.45\textwidth}
\centering
\includegraphics[width=\textwidth]{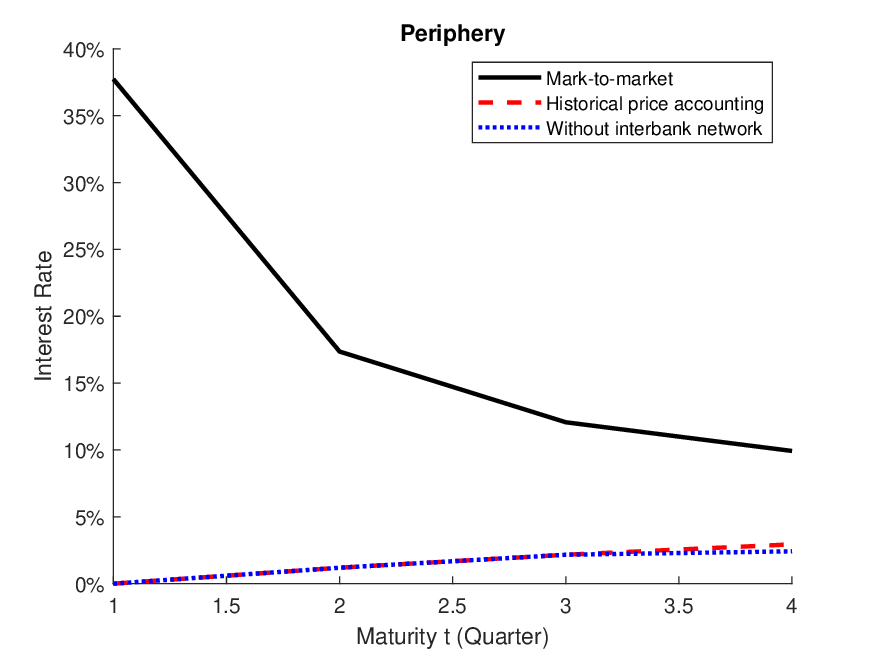}
\caption{Periphery: High volatility setting}
\label{fig:cp-periphery-stressed}
\end{subfigure}
\caption{Section~\ref{sec:Kmaturity-cs-cp}: The term structure for the core and peripheral banks under the two stress scenarios (i.e., $\sigma_C^2 = 0.5$ and $\sigma_C^2 = 0.75$ respectively) and following varying accounting systems.}
\label{fig:cp}
\end{figure}

As depicted in Figures~\ref{fig:cp-core-unstressed} and~\ref{fig:cp-periphery-unstressed}, in the low volatility setting, all banks in the system have a normal yield curve with interest rates $\vR^*(t_l)$ in the low single digits at each maturity date. Furthermore, the choice of accounting rules has only marginal effects on the term structure for both core and peripheral banks.  In comparison, in the high volatility setting displayed in Figures~\ref{fig:cp-core-stressed} and~\ref{fig:cp-periphery-stressed}, all banks have an inverted yield curve with interest rates in the double digits when considering the mark-to-market valuation. Notably, the stress scenario only includes an increase in the core volatilities $\sigma_C^2$ without any changes to the current balance sheets and with no direct change to the dynamics of the peripheral banks.  The change in the term structure for peripheral banks is being driven entirely by informational contagion from the potential future distress of the core institutions. In particular, when we neglect this channel of distress contagion, through either the use of historical price accounting or the no-interbank network systems, then the associated term structures for both the core and peripheral institutions are almost identical between the stressed and unstressed scenarios. Though a stylized system, this underscores the importance of regulating systemically important financial institutions. When incorporating contagion from worries about their future states, increased levels of volatility for a core set of banks alone may be amplified and propagate readily throughout the entire financial system, causing distress in the short-term even if there are no immediate actualized defaults.

\section{Conclusion}

The recent failure of SVB and the associated banking crisis has highlighted just how rapidly and broadly concerns about the \emph{future} standing of one or more banks can spread to other institutions, causing distress across the financial system \emph{today}, with possible defaults to follow.

Similarly, worries about counterparty defaults involving major financial institutions are believed to have played a key role in the unfolding of the global financial crisis (\cite{claessens_kodres}). Later works have proposed corresponding models of informational contagion, via some form of equilibrium mark-to-market valuation of interbank claims, with \cite{glasserman2016contagion} stressing that, in their view, \emph{``this loss-of-confidence effect [...] is one of the most important channels through which the financial network amplifies systemic risk in practice''.}

Herein, we have extended these mark-to-market approaches (\cite{fischer2014, barucca2020network, bardoscia_2019, banerjee2022pricing}) to a genuinely dynamic framework with multiple due dates of the interbank obligations and a canonical mechanism for their equilibrium valuation; the latter being based on the current risk of future defaults owing either to illiquidity or insolvency which may occur at any point in time, as determined by the evolution of banks' capital and cash accounts. Thus, our framework is able to address how concerns about potential future contagion can shape the current dynamics of the financial system.

Leveraging this, we have shown how distress contagion can translate into significant excess volatility in down market scenarios, and we have derived intrinsic risk-aware term structures for interbank obligations that can reveal how systemic risk is concentrated across different windows of time. In line with recent empirical work of \cite{Bluwstein_2023}, we have demonstrated that non-negligible levels of systemic risk are naturally reflected in inverted shapes of these term structures. If one instead relies on historical price accounting, as in traditional approaches to default contagion, then we have seen that such effects may not be picked up, and one would in general be in danger of underestimating both the timing and severity of systemic risk.

The above is evocative of an ongoing debate ignited by the global financial crisis, namely that of the role played by mark-to-market rules in amplifying the crisis (\cite{laux_leuz1}) as well as the wider implications of accounting and reporting practices for financial stability (\cite{bischof_laux_leuz}). Crucially, our framework is not intended to reflect the practices that are actually followed:~different rules apply depending on asset classes, e.g.~enabling non-marketable interbank exposures to be carried at historical cost and allowing for deviations from the enforcement of fair values in times of stress. Rather, it is a framework for what \emph{could} happen if worries about fair values begin to propagate, and, accordingly, it may inform the development of supervisory tools to asses a given network's exposure to this form of systemic risk. As discussed in \cite{laux_leuz2}, even if all assets were carried at historical cost during the global financial crisis, and hence financial institutions may have looked in better shape, outside investors and regulators would have been concerned about the true quality of their balance sheets.

Naturally, there are trade-offs involved. \cite{laux_leuz2} conclude that, \emph{``on one hand, marking assets to market prices can in principle exacerbate downward spirals and contagion during a financial crisis; but on the other hand, a faster recognition of losses provides pressures for prompt corrective action by banks and regulators and likely limits imprudent lending in the first place.''} Our framework provides a principled approach to assessing why and when such corrective action is appropriate. Of course, if regulators and banking supervisors were to define and enforce regulatory constraints (e.g., capital requirements) based on this, they would need to tread carefully in relation to the potential (self-fulfilling) shortcomings of making it easier for banks to violate the constraints. Nevertheless, markets will ultimately look for the fair value of assets and the knowledge that such a framework is enforced could proactively affect internal risk management principles. The current practice of applying prudential filters, relaxing constraints away from mark-to-market rules, could be having the opposite effect, and it is at risk of overlooking a critical channel for the rapid manifestation of financial crises.

In order to apply this methodology in practice, the model needs to be brought to financial data. Due to the scaling of the multinomial tree-based approach taken herein, the computational requirements grow exponentially in the number of time steps. To overcome this issue for larger financial systems, appropriate Monte Carlo or machine learning approaches would need to be developed so as to find approximate clearing solutions. We leave a discussion of such numerical approaches to future work.

\subsection*{Acknowledgments}\vspace{-6pt}
We are grateful for the valuable comments and suggestions from two anonymous referees as well as the Editor and the Associate Editor. Part of the manuscript was written when ZF visited Vienna University of Economics and Business. In this regard, he is thankful for the support of the OeNB anniversary fund, project number 17793. Further, AS gratefully acknowledges the hospitality of the Institute for Mathematical and Statistical Innovation, where he was a visitor while working on the early stages of this project, as part of the long program on Decision Making and Uncertainty.

\subsection*{Conflict of interest statement}\vspace{-6pt}
The authors have no conflicts of interest to declare.

\subsection*{Data availability statement}\vspace{-6pt}
Data sharing is not applicable to this article as no datasets were generated or analysed during the current study.

\bibliographystyle{plainnat}
\bibliography{bibtex2}

\newpage
\appendix

\section{Marking-to-market in related static frameworks}\label{sect:existing_works}

As mentioned in the introduction, \cite{glasserman2016contagion} discuss how mark-to-market valuation within an interbank network can provide a natural approach to informational contagion, allowing one to capture losses in confidence concerning the creditworthiness of individual banks. Starting from a given (static) balance sheet with nominal values $L_{ij}$, mark-to-market values $0\leq p_{ij} \leq L_{ij}$ (in our notation) are introduced by positing a valuation function $\phi_j : \bbr^n_+ \rightarrow  \bbr^n_+$ that maps bank $j$'s assets to the value of its obligations. It is assumed that $0\leq \phi_j(v) \leq L_{j \cdot }$ for $0\leq v \leq L_{ \cdot j}$ and that $\phi_j$ is non-decreasing. Tarski's fixed point theorem then gives existence of an $n\times n$ matrix of the desired mark-to-market values $0\leq p_{ij} \leq L_{ij}$ satisfying $p_{j \cdot }=\phi_j(p_{\cdot j})$ for $j=1,\ldots n $. Our framework may be seen as a dynamic extension of this equilibrium problem with a canonical mechanism for the current (now stochastic) mark-to-market valuations based on worries about possible future defaults and, correspondingly, worries about how strongly this could translate into distress and default contagion dynamically over time.

It is worth noting that such a framework is not confined to interbank networks. For example, \cite{paddrik2020contagion} have adapted the mark-to-market approach of \cite{glasserman2016contagion} to study variation margins in a network of CDS contracts. First, they introduce a notion of `stress', defined as the nominal amount by which outgoing payment obligations exceed incoming ones (for given spreads that will then be shocked) accounting also for initial margins and liquidity buffers (all of which may be viewed as the components of a balance sheet for each firm). Then, an ad-hoc valuation function (called the `stress response function') is applied to get a reduced value for each bank's `expected' payments to its counterparties. Finally, one then solves for the equilibrium size of these `expected' payments, as with the mark-to-market values above.

Inspired by the Black--Cox model (\cite{black1976valuing}) for a single firm, \cite{bardoscia_2019} have proposed what they refer to as a model of forward-looking solvency contagion, wherein the realized capital at the current time $t$ is influenced by the possibility of insolvency before or at a given terminal time $T$, when obligations are due. In their setup, using our notation, the capital $K$ at the current time $t$ can be seen to satisfy
\begin{align}\label{eq:bardoscia_capital}
	K_i(t) &= x_i(t) + e^{-r(T-t)}\sum_{j = 1}^n L_{ji} \bigl( \beta+(1-\beta)\widetilde{P}_j(t)\bigr) - e^{-r(T-t)}\bar p_i, \\
	\widetilde{P}_i(t) &= \P(\widetilde{\tau}_i > T \mid \vx(t)), \quad \widetilde{\tau}_j := \inf \{ s \geq t : x_i(s) + e^{r(s-t)}(K_i(t) - x_i(t)) \leq 0  \},\nonumber
\end{align}
where the external assets $\vx(t):=(x_1(t),\ldots,x_n(t))$ at time $t$ are given and each $x_i$ is modelled by a geometric Brownian motion on $[t,T]$. This clearing problem solely involves the current capital at time $t$. In particular, it is \emph{not} a dynamic model of contagion, as the $\widetilde{P}_i$'s disregard expected future interactions. Moreover, the capital $K$ in \eqref{eq:bardoscia_capital} is not well-defined at the stochastic process level, for $t\in[0,T]$, as the $\widetilde{P}_i$'s do not account for past information (e.g., past defaults); it is only meaningful in a one-period sense. That is, \eqref{eq:bardoscia_capital} is a one-period model taking place at a given time $t$, using geometric Brownian motion and the single-firm default rule of \cite{black1976valuing} with terminal time $T-t$ to construct particular valuation functions corresponding to the framework of \cite{glasserman2016contagion}. Indeed, setting $\widetilde{p}_{ij} := e^{-r(T-t)} ( \beta+(1-\beta)\widetilde{P}_i(t)) L_{ij}$, we see that $\widetilde{P}_i(t)=\psi_i(\widetilde{p}_{\cdot i})$, where each $\psi_i$ can be computed explicitly, precisely as in \cite{black1976valuing}. In this way, \eqref{eq:bardoscia_capital} reduces to the fixed point problem
\[
\widetilde{p}_{i \hspace{0.5pt} \cdot }=\phi_i(\widetilde{p}_{\cdot i}),\quad \phi_i(v):= e^{-r(T-t)} ( \beta+(1-\beta)\psi_i(v)) L_{i\hspace{0.5pt} \cdot},\quad i=1,\ldots,n,
\]
where each $\phi_i : \bbr^n_+ \rightarrow  \bbr^n_+ $ is given explicitly, satisfies $0\leq \phi_i\leq L_{i \hspace{0.5pt} \cdot }$, and is non-decreasing.

\section{Comparison with historical price accounting}\label{sec:hpa}
Throughout this paper, we followed a mark-to-market accounting principle for valuing interbank assets. This is in contrast to historical price accounting, which lies at the core of existing dynamic and stochastic frameworks for default contagion in interbank networks, see \cite{Lipton2016, feinstein2021dynamic, BBF18} and the related mean-field models of \cite{NS17,HLS18, FS23}. In these works, interbank assets are valued in full until a default event is actualized: only at the default time, the value of any defaulting assets are then re-marked to the realized payout value. In this appendix, we derive the historical price accounting analogue of our model and compare the clearing solutions. Specifically, we provide conditions so that the clearing solutions for historical price accounting always outperform the mark-to-market accounting rules for every bank.

\subsection{Single maturity setting}\label{sec:hpa-1}
In considering historical price accounting, the only modification to the clearing system that needs to be considered is in the accounting of the interbank assets. In the single maturity setting of Section~\ref{sec:1maturity}, this requires the alteration of $\Psi^{T}_{\vP}$ to only consider actualized defaults.  Mathematically define $\Psi^{H,T}_{\vP}: \{t_0,t_1,...,t_{\ell}\} \times \{t_0,t_1,...,t_{\ell},T+1\}^{|\Omega| \times n} \to [\vec{0},\vec{1}]^{\ell + 1}$ by
\[\Psi^{H,T}_{\vP,i}(t,\vt) := \ind{\tau_i > t}\]
for any bank $i = 1,...,n$, time $t$, and stopping (default) times $\vt$.

With this modification to the accounting of interbank assets, $(\vK,\vP,\vt) \in \bbd^T$ is a clearing solution if it is a fixed point for $\Psi^{H,T} = (\Psi^T_{\vK},\Psi^{H,T}_{\vP},\Psi^T_{\vt})$, i.e.,
\begin{align}
	\label{eq:1maturity-hpa} &(\vK,\vP,\vt) = \Psi^{H,T}(\vK,\vP,\vt) := (\Psi^T_{\vK}(t_l,\vP(t_l)) \; , \; \Psi^{H,T}_{\vP}(t_l,\vt) \; , \; \Psi^T_{\vt}(\vK))_{l = 0}^\ell. 
\end{align}
\begin{proposition}\label{prop:hpa-1maturity-exist}
	The set of clearing solutions to~\eqref{eq:1maturity-hpa}, i.e., all elements $(\vK^*,\vP^*,\vt^*) \in \bbd^T $ such that $(\vK^*,\vP^*,\vt^*) = \Psi^{H,T}(\vK^*,\vP^*,\vt^*)$, forms a lattice in $\bbd^T$ with greatest and least solutions  $(\vK^\uparrow,\vP^\uparrow,\vt^\uparrow) \geq (\vK^\downarrow,\vP^\downarrow,\vt^\downarrow)$.
\end{proposition}
\begin{proof}
	As with Proposition~\ref{prop:GK} and Theorem~\ref{thm:1maturity-exist}, this result follows from a direct application of Tarski's fixed point theorem since $\Psi^{H,T}$ is monotonic in the complete lattice $\bbd^T$. 
\end{proof}

We conclude our discussion of the single maturity setting by directly comparing the clearing solutions found under the mark-to-market setting of~\eqref{eq:1maturity} and the historical price accounting setting of~\eqref{eq:1maturity-hpa}. 
\begin{lemma}\label{lemma:mtm-hpa-1maturity}
	For any mark-to-market clearing solution $(\vK^*,\vP^*,\vt^*) = \Psi^T(\vK^*,\vP^*,\vt^*)$, there exists a historical price accounting clearing solution $(\vK^{H,*},\vP^{H,*},\vt^{H,*}) = \Psi^{H,T}(\vK^{H,*},\vP^{H,*},\vt^{H,*})$ that bounds $(\vK^*,\vP^*,\vt^*)$ from above, i.e., $(\vK^*,\vP^*,\vt^*) \leq (\vK^{H,*},\vP^{H,*},\vt^{H,*})$.
	Conversely, for any historical price accounting clearing solution $(\vK^{H,*},\vP^{H,*},\vt^{H,*}) = \Psi^{H,T}(\vK^{H,*},\vP^{H,*},\vt^{H,*})$, there exists a mark-to-market clearing solution $(\vK^*,\vP^*,\vt^*) = \Psi^T(\vK^*,\vP^*,\vt^*)$ that bounds $(\vK^{H,*},\vP^{H,*},\vt^{H,*})$ from below, i.e., $(\vK^*,\vP^*,\vt^*) \leq (\vK^{H,*},\vP^{H,*},\vt^{H,*})$.
\end{lemma}
\begin{proof}
	First, by construction, $\Psi^{H,T}(\vK,\vP,\vt) \geq \Psi^T(\vK,\vP,\vt)$ for any $(\vK,\vP,\vt) \in \bbd^T$. 
	
	Now let $(\vK^*,\vP^*,\vt^*) = \Psi^T(\vK^*,\vP^*,\vt^*) \leq \Psi^{H,T}(\vK^*,\vP^*,\vt^*)$. Recall that $\bbd^T$ is a complete lattice, therefore the limit of Picard iterations over $\Psi^{H,T}$ exists and is a fixed point. Due to monotonicity of $\Psi^{H,T}$ if immediately follows that this fixed point is bounded from below by $(\vK^*,\vP^*,\vt^*)$.
	
	Finally, let $(\vK^{H,*},\vP^{H,*},\vt^{H,*}) = \Psi^{H,T}(\vK^{H,*},\vP^{H,*},\vt^{H,*}) \geq \Psi^T(\vK^{H,*},\vP^{H,*},\vt^{H,*})$. By a similar argument as before, we can realize a fixed point of $\Psi^T$ through Picard iterations that is bounded from above by $(\vK^{H,*},\vP^{H,*},\vt^{H,*})$.
\end{proof}

\begin{remark}\label{rem:1maturity-comparison}
	The statement of Lemma~\ref{lemma:mtm-hpa-1maturity} is written generally for any clearing solutions. Notably these imply that the monotonicity holds for the greatest and least clearing solutions, e.g., $(\vK^\uparrow,\vP^\uparrow,\vt^\uparrow) \leq (\vK^{H,\uparrow},\vP^{H,\uparrow},\vt^{H,\uparrow})$ so that the greatest clearing solution under historical price accounting is uniformly better for all banks than mark-to-market accounting.
	
	In this way, one can view historical price accounting as the optimistic accounting rule while mark-to-market is the pessimistic one. From a regulatory perspective, if the system is healthy under mark-to-market accounting then it can reasonably be considered resilient to shocks as this is a ``worst case'' for fair accounting rules. 
In fact, the liquidity-only clearing system (presented in Remark~\ref{rem:illiquid}) lead to an even more optimistic accounting than the historical price accounting rule. Recalling that the liquidity-only rules relate to the static models of, e.g.,~\cite{GK10}, we find that such a system is overly optimistic compared to any dynamic accounting of system health, i.e., $(\vK^\uparrow,\vP^\uparrow,\vt^\uparrow) \leq (\vK^{H,\uparrow},\vP^{H,\uparrow},\vt^{H,\uparrow}) \leq (\vK^{L,\uparrow},\vP^{L,\uparrow},\vt^{L,\uparrow})$.
\end{remark}

\subsection{Multiple maturity setting}\label{sec:hpa-K}
Similar to the single maturity setting, in constructing a clearing solution under historical price accounting we need only alter our treatment of the future interbank assets.  That is, we consider a clearing solution via the fixed point problem
\begin{align}
	\label{eq:Kmaturity-hpa} &(\vK,\vV,\vP,\vt) = \Psi^H(\vK,\vV,\vP,\vt;\va) \\
	\nonumber    &\;\; := (\Psi_{\vK}(t_l,\vV(t_l),\vP(t_l,\cdot),\vt) \; , \; \Psi_{\vV}(t_l,\vV(t_{[l-1]\vee 0}),\vt;\va) \; , \; \Psi_{\vP}^H(t_l,t_k,\vt)_{k = [l+1]\wedge\ell}^{\ell} \; , \; \Psi_{\vt}(\vK,\vV))_{l = 0}^{\ell}
\end{align}
where, for any $i = 1,...,n$,
\[\Psi_{\vP,i}^H(t_l,t_k,\vt) = \ind{\tau_i > t_l}.\]
We wish to note that $\Psi_{\vP,i}^H(t_l,t_k,\vt)$ is independent of $t_k$; we leave the dependence here for easier comparisons to $\Psi_{\vP,i}$.
\begin{proposition}\label{prop:Kmaturity-exist-hpa}
	Fix the rebalancing strategies $\va(t_l,\hat\vK,\hat\vV) \in [\vec{0},\vec{1}]^{|\Omega_{t_l}|}$ so that they depend only on the current time $t_l$, capital $\hat\vK \in \lcal^n_{t_l}$, and cash account $\hat\vV \in \lcal_{t_l}^n$. There exists a (finite) clearing solution $(\vK^*,\vV^*,\vP^*,\vt^*) = \Psi^H(\vK^*,\vV^*,\vP^*,\vt^*)$ to \eqref{eq:Kmaturity-hpa}. Furthermore, if we have Markovian external assets, $\vx(t_l) = f(\vx(t_{l-1}),\tilde\epsilon(t_l))$ for i.i.d.\ perturbations $\tilde\epsilon$, then there exists a clearing solution such that $(\vx(t_l),\vK^*(t_l),\vV^*(t_l),\vP^*(t_l,t_k)_{k = l}^\ell,\vi^*(t_l))_{l = 0}^\ell$ is Markovian where $\vi^*(t_l) := \ind{\vt^* \geq t_l}$ is the realized solvency process.
\end{proposition}
\begin{proof}
	As with the proof of Theorem~\ref{thm:Kmaturity-exist}, this result can be proven constructively. In particular, due to the use of historical price accounting, this clearing solution can be constructed forward in time as there are no forward dependencies necessary. Each time can be formulated via a fixed point problem with existence via Tarski's fixed point theorem, so the desired results follow.
\end{proof}

We conclude this discussion by considering the special case where a direct comparison can be drawn between the mark-to-market setting \eqref{eq:Kmaturity} and the historical price accounting setting \eqref{eq:Kmaturity-hpa}. Notably, the strong monotonicity properties that allow mark-to-market accounting to serve as a reliable stress testing device in the single maturity setting (as discussed in Remark~\ref{rem:1maturity-comparison}), is only guaranteed to hold under the zero recovery setting $\beta = 0$.
\begin{lemma}\label{lemma:mtm-hpa-Kmaturity}
	Consider a financial system with zero recovery rate ($\beta = 0$) and such that all rebalancing strategies are independent of bank performance (e.g., $\va \equiv \vec{0}$).
	For any mark-to-market clearing solution $(\vK^*,\vV^*,\vP^*,\vt^*) = \Psi(\vK^*,\vV^*,\vP^*,\vt^*)$, there exists a historical price accounting clearing solution $(\vK^{H,*},\vV^{H,*},\vP^{H,*},\vt^{H,*}) = \Psi^H(\vK^{H,*},\vV^{H,*},\vP^{H,*},\vt^{H,*})$ that bounds $(\vK^*,\vV^*,\vP^*,\vt^*)$ from above, i.e., $(\vK^*,\vV^*,\vP^*,\vt^*) \leq (\vK^{H,*},\vV^{H,*},\vP^{H,*},\vt^{H,*})$.
	
	Conversely, for any historical price accounting clearing solution $(\vK^{H,*},\vV^{H,*},\vP^{H,*},\vt^{H,*}) = \Psi^H(\vK^{H,*},\vV^{H,*},\vP^{H,*},\vt^{H,*})$, there exists a mark-to-market clearing solution $(\vK^*,\vV^*,\vP^*,\vt^*) = \Psi(\vK^*,\vV^*,\vP^*,\vt^*)$ that bounds $(\vK^{H,*},\vV^{H,*},\vP^{H,*},\vt^{H,*})$ from below, i.e., $(\vK^*,\vV^*,\vP^*,\vt^*) \leq (\vK^{H,*},\vV^{H,*},\vP^{H,*},\vt^{H,*})$.
\end{lemma}
\begin{proof}
	Under the assumed properties both $\Psi$ and $\Psi^H$ are monotonic mappings in the complete lattice $\bbd$. Furthermore, by construction, $\Psi^H(\vK,\vV,\vP,\vt) \geq \Psi(\vK,\vV,\vP,\vt)$ for any $(\vK,\vV,\vP,\vt) \in \bbd$. Therefore, by a comparable comparison principle argument as in the proof of Lemma~\ref{lemma:mtm-hpa-1maturity}, the desired results follow. 
\end{proof}

\begin{remark}
	Though a local comparison can be made between the mark-to-market and historical price accounting systems, the total comparison as provided within Lemma~\ref{lemma:mtm-hpa-Kmaturity} fails in the general case due to: (i) the positive impact defaults have on short-term liquidity when $\beta > 0$ and (ii) the variable impact that rebalancing has on both the capital and cash accounts.
\end{remark}

\section{Proofs}\label{sec:proofs}
\subsection{Proofs for the single maturity setting of Section~\ref{sec:1maturity}}\label{sec:proofs-1maturity}
\subsubsection{Proof of Theorem~\ref{thm:1maturity-exist}}
\begin{proof}
The result follows from a direct application of Tarski's fixed point theorem since $\Psi^T$ is monotonic in the complete lattice $\bbd^T$. 
\end{proof}

\subsubsection{Proof of Proposition \ref{prop:volatility}}
\begin{proof}
	In the notation of Section \ref{sec:setting-tree}, at any time $t \in \{0,\dt,\ldots, T-\dt\}$ and for each bank $k=1,\ldots,n$, we can write
	\[
	\tilde{x}_k(t+\dt , \omega^n_{t+\dt , (n+1)(i-1)+j}) = \hat{x}_k^{(j)}(t,\omega^n_{t,i}),
	\]
	for a given function $\hat{x}_k^{(j)}$, for every $i=1,\ldots, (n+1)^{t/\dt}$ and $j=1,\ldots,n+1$. Here $i$ corresponds to the node reached at time $t$ and $j$ corresponds to the branch taken to time $t+\dt$. Similarly, we can write
	\[
	P^{(j)}_k(t+\dt,\omega^n_{t+\dt , (n+1)(i-1)+j} )=\hat{P}^{(j)}_k(t ,\omega^n_{t,i}).
	\]
	By construction, for any given event $\omega^n_{t,i} \in \fcal_t^n$, the vectors $(\hat{x}^{(1)}_k(t,\omega^{n}_{t,i}),\ldots,\hat{x}^{(n+1)}_k(t,\omega^{n}_{t,i}))$, for $k=1\ldots,n$, and the vector $(1,\ldots,1)$ are linearly independent in $\bbr^{n+1}$. Write $\hat{\vx}^{(j)}=(\hat{x}^{(j)}_1,\ldots, \hat{x}^{(j)}_n)$ and $\hat{\vP}^{(j)}=(\hat{P}^{(j)}_1,\ldots,\hat{P}^{(j)}_n)$ for $j=1,\ldots,n+1$.  For $t\in\{0,\dt, \ldots , T-\dt\}$, the aforementioned linear independence allows us to define
	\begin{equation}\label{eq:theta}
		\boldsymbol{\theta}(t , \omega):=\left(\begin{array}{c}
			\hat{\vx}^{(2)}(t,\omega_{t,i}^n)-\hat{\vx}^{(1)}(\omega^n_{t,i})\\
			\vdots\\
			\hat{\vx}^{(n+1)}(t,\omega_{t,i}^n)-\hat{\vx}^{(1)}(t,\omega^n_{t,i})
		\end{array}\right)^{\!-1}\left(\begin{array}{c}
			\hat{\vP}^{(2)}(t,\omega_{t,i}^n)-\hat{\vP}^{(1)}(t,\omega_{t,i}^n)\\
			\vdots\\
			\hat{\vP}^{(n+1)}(t,\omega_{t,i}^n)-\hat{\vP}^{(1)}(t,\omega_{t,i}^n)
		\end{array}\right),
	\end{equation}
	whenever $\omega \in \succ(\omega^n_{t,i})$, where we note that the right-hand side of \eqref{eq:theta} is $\fcal_t^n$-measurable. It now follows from \eqref{eq:theta} that, for any given $\omega^n_{t,i}\in \fcal^n_t$ and $\omega \in \succ(\omega^n_{t,i})$, we have
	\[
	\hat{\vP}^{(j)}(t,\omega^n_{t,i})- (\hat{\vx}^{(j)}(t,\omega^n_{t,i})-\vx(t,\omega^n_{t,i}))	\boldsymbol{\theta}(t,\omega) =\hat{\vP}^{(1)}(t,\omega^n_{t,i})- (\hat{\vx}^{(1)}(t,\omega^n_{t,i})-\vx(t,\omega^n_{t,i}))	\boldsymbol{\theta}(t,\omega),
	\]
	for each $j=1,\ldots,n+1$. Consequently, $\vP(t+\dt) - \Delta \tilde{\vx}(t) 	\boldsymbol{\theta}(t)$ is known conditional on $\fcal^n_{t}$. Since $\tilde{\vx}$ is a martingale and $\theta(t)$ is $\fcal^n_t$-measurable, this implies that $\vP(t+\dt) -  \Delta \tilde{\vx}(t) 	\boldsymbol{\theta}(t)$ equals $\E[\vP(t+\dt) \,|\,\fcal^n_{t}]$, and hence the martingale property of $\vP$ gives
	\[
	\Delta P_k(t)=  P_k(t + \dt) - \E[P_k(t+\dt) \,|\, \fcal^n_{t}] = 	\boldsymbol{\theta}_k(t)^\T \Delta \tilde{\vx}(t),
	\]
	for each $k=1,\ldots,n$, where $\boldsymbol{\theta}_k$ denotes the $k^\text{th}$ column of $\boldsymbol{\theta}$. The claim now follows from $\vK(t) = \Psi^T_{\vK}(t,\vP(t))$ in \eqref{eq:1maturity}, by inserting the above expression for $\Delta P_k(t)$.
\end{proof}

\subsubsection{Proof of Proposition~\ref{prop:1maturity-recursion}}
\begin{proof}
Let $(\vK,\vP,\vt) \in \bbd^T$ be a fixed point of~\eqref{eq:1maturity}.  This is a fixed point of~\eqref{eq:1maturity-recursion} if and only if $P_i(t_l,\omega_{t_l}) = \bar\Psi_{\vP,i}^T(t_l,\vP(t_{[l+1]\wedge\ell}),\vt,\omega_{t_l})$.
\begin{itemize}
\item At $l = \ell$: $P_i(T,\omega_T) = \P(\tau_i > T | \fcal_T)(\omega_T) = \ind{\tau_i(\omega_T) > T}$ for every $\omega_T \in \Omega_T$ by construction of $\fcal_T = \fcal$.
\item At $l < \ell$: 
    Fix $\omega_{t_l} \in \Omega_{t_l}$,
    \begin{align*}
    P_i(t_l,\omega_{t_l}) &= \P(\tau_i > T | \fcal_{t_l})(\omega_{t_l}) = \frac{\P(\tau_i > T , \omega_{t_l})}{\P(\omega_{t_l})} = \sum_{\omega_{t_{l+1}} \in \succ(\omega_{t_l})} \frac{\P(\tau_i > T , \omega_{t_{l+1}})}{\P(\omega_{t_l})} \\
    &= \sum_{\omega_{t_{l+1}} \in \succ(\omega_{t_l})} \frac{\P(\omega_{t_{l+1}}) \P(\tau_i > T | \fcal_{t_{l+1}})(\omega_{t_{l+1}})}{\P(\omega_{t_l})} = \sum_{\omega_{t_{l+1}} \in \succ(\omega_{t_l})} \frac{\P(\omega_{t_{l+1}}) P(t_{l+1},\omega_{t_{l+1}})}{\P(\omega_{t_l})}.
    \end{align*}
\end{itemize}

Let $(\vK,\vP,\vt) \in \bbd^T$ be a fixed point of~\eqref{eq:1maturity-recursion}.  This is a fixed point of~\eqref{eq:1maturity} if and only if $P_i(t_l) = \Psi_{\vP,i}^T(t_l,\vt)$ almost surely very every time $l = 0,1,...,\ell$.
\begin{itemize}
\item At $l = \ell$: $P_i(T,\omega_T) = \ind{\tau_i(\omega_T) > T} = \P(\tau_i > T | \fcal_T)(\omega_T)$ for every $\omega_T \in \Omega_T$ by construction of $\fcal_T = \fcal$.
\item At $l < \ell$: Assume $P_i(t_{l+1}) = \Psi_{\vP,i}^T(t_{l+1},\vt)$ at time $t_{l+1}$ almost surely. Fix $\omega_{t_l} \in \Omega_{t_l}$,
    \begin{align*}
    P_i(t_l,\omega_{t_l}) &= \sum_{\omega_{t_{l+1}} \in \succ(\omega_{t_l})} \frac{\P(\omega_{t_{l+1}}) P(t_{l+1},\omega_{t_{l+1}})}{\P(\omega_{t_l})} = \sum_{\omega_{t_{l+1}} \in \succ(\omega_{t_l})} \frac{\P(\omega_{t_{l+1}}) \P(\tau_i > T | \fcal_{t_{l+1}})(\omega_{t_{l+1}})}{\P(\omega_{t_l})} \\
    &= \sum_{\omega_{t_{l+1}} \in \succ(\omega_{t_l})} \frac{\P(\tau_i > T , \omega_{t_{l+1}})}{\P(\omega_{t_l})} = \frac{\P(\tau_i > T , \omega_{t_l})}{\P(\omega_{t_l})} = \P(\tau_i > T | \fcal_{t_l})(\omega_{t_l}).
    \end{align*}
\end{itemize}
\end{proof}

\subsubsection{Proof of Proposition~\ref{prop:1maturity-dpp-define}}
\begin{proof}
As with Theorem~\ref{thm:1maturity-exist}, this result follows readily from Tarski's fixed point theorem.
\end{proof}

\subsubsection{Proof of Proposition~\ref{prop:1maturity-dpp}}
\begin{proof}
Define $(\vK,\vP)$ to be the realized solution from $\hat\Psi^T(0,\vec{1})$.  Let $\vt := \Psi_{\vt}^T(\vK)$ be the associated default times and $\vi(t_l,\omega_{t_l}) := \prod_{k = 0}^{l-1} \ind{\vK(t_k,\omega_{t_l}) \geq \vec{0}}$ be the realized (auxiliary) solvency process at time $t$ and in state $\omega_{t_l} \in \Omega_{t_l}$ (and $\vi(0,\omega_0) = \vec{1}$).  
(We wish to note that $\vi(t_l) \in \lcal_{t_{[l-1]^+}}^n$ and as such could be indexed by the preceding states $\omega_{t_{[l-1]^+}} \in \Omega_{t_{[l-1]^+}}$ instead; we leave the use of $\omega_{t_l}$ as we find it is clearer notationally.)
First, we will show that $(\vK,\vP,\vt)$ is a clearing solution of~\eqref{eq:1maturity} via the representation~\eqref{eq:1maturity-recursion}, i.e., $(\vK,\vP,\vt) = \bar\Psi^T(\vK,\vP,\vt)$. 
Second, we will show that this solution must be the maximal clearing solution as proven to exist in Theorem~\ref{thm:1maturity-exist}.

\begin{enumerate}
\item By construction of $\hat\Psi^T$, $(\vK,\vP,\vt) = \bar\Psi^T(\vK,\vP,\vt)$ if and only if $\vP(t_l) = \bar\Psi^T_{\vP}(t_l,\vP(t_{[l+1]\wedge\ell}),\vt)$ for every time $l \in \{0,1,...,\ell\}$.
At maturity, $\vP(T) = \bar\Psi^T_{\vP}(T,\vP(T),\vt)$ trivially by construction of $\vt$.
Consider now $l < \ell$ and assume $\vP(t_{l+1}) = \hat\Psi^T_{\vP}(t_{l+1},\iota(t_{l+1}))= \bar\Psi^T_{\vP}(t_l,\vP(t_{l+1}),\vt)$.   By construction, $\vP(t_l,\omega_{t_l}) = \bar\Psi^T_{\vP}(t_{l+1},\vP(t_{l+1}),\vt,\omega_{t_l})$ and the result is proven.  
\item Now assume there exists some clearing solution $(\vK^\dagger,\vP^\dagger,\vt^\dagger) \gneq (\vK,\vP,\vt)$. Then we can rewrite the form of $(\vK^\dagger,\vP^\dagger)$ as:
\begin{align*}  
(\vK^\dagger,\vP^\dagger) &= (\Psi^T_{\vK}(t_l,\vP^\dagger(t_l)) , \bar\Psi^T_{\vP}(t_l,\vP^\dagger(t_{[l+1]\wedge\ell}),\Psi^T_{\tau}(\vK^\dagger)))_{l = 0}^\ell
\end{align*}
through the use of the clearing formulation $\bar\Psi^T$ and explicitly applying $\vt^\dagger = \Psi^T_{\vt}(\vK^\dagger)$.  Following the logic of the prior section of this proof, it must follow that
\begin{equation*}
\vP^\dagger(t_l,\omega_{t_l}) = \begin{cases} \diag{\ind{\inf_{k < l} \vK(t_k,\omega_{t_l}) \geq \vec{0}}} \sum_{\omega_{t_{l+1}} \in \succ(\omega_{t_l})} \frac{\P(\omega_{t_{l+1}}) \vP(t_{l+1},\omega_{t_{l+1}})}{\P(\omega_{t_l})} &\text{if } l < \ell \\ \ind{\inf_{k \in [0,\ell]} \vK(t_k,\omega_T) \geq \vec{0}} &\text{if } l = \ell \end{cases}
\end{equation*}
for every time $t_l$ and state $\omega_{t_l} \in \Omega_{t_l}$.
That is, $(\vK^\dagger(t_l),\vP^\dagger(t_l)) = \hat\Psi^T(t_l,\ind{\inf_{k < l} \vK(t_k,\omega_{t_l}) \geq \vec{0}})_{\omega_{t_l} \in \Omega_{t_l}}$ satisfies all of the fixed point problems within the construction of $\hat\Psi^T$ at all times $t_l$.  
We will complete this proof via backwards induction with, to simplify notation, $\vi(t_l) = \ind{\inf_{k < l} \vK^\dagger(t_k) \geq \vec{0}}$.
Consider maturity $T$, it must follow that $(\vK^\dagger(T),\vP^\dagger(T)) \leq \hat\Psi^T(T,\vi(T))$ by the definition of the fixed point operator $\FIX$.
Consider some time $t_l < T$ and assume $(\vK^\dagger(t_{l+1}),\vP^\dagger(t_{l+1})) \leq \hat\Psi^T(t_{l+1},\vi(t_{l+1}))$.  By the backward recursion used within $\hat\Psi^T_{\vP}(t_l,\vi(t_l))$, it follows that
\begin{align*}
\vP^\dagger(t_l,\vi(t_l)) &= \diag{\vi(t_l)}\left[\sum_{\omega_{t_{l+1}} \in \succ(\omega_{t_l})} \frac{\P(\omega_{t_{l+1}})\vP(t_{l+1},\omega_{t_{l+1}})}{\P(\omega_{t_l})}\right]_{\omega_{t_l} \in \Omega_{t_l}} \\
&\leq \diag{\vi(t_l)}\left[\sum_{\omega_{t_{l+1}} \in \succ(\omega_{t_l})} \frac{\P(\omega_{t_{l+1}})\hat\Psi^T_{\vP}(t_{l+1},\vi(t_{l+1},\omega_{t_{l+1}}),\omega_{t_{l+1}})}{\P(\omega_{t_l})}\right]_{\omega_{t_l} \in \Omega_{t_l}}.
\end{align*}
Further, by the monotonicity of $\Psi^T_{\vK}$, it would follow that $(\vK^\dagger(t_l),\vP^\dagger(t_l)) \leq \hat\Psi^T(t_l,\vi(t_l))$.
This, together with the trivial monotonicity of $\hat\Psi^T$ w.r.t.\ the solvency indicator $\vi$, forms a contradiction to the original assumption. 
\end{enumerate}
\end{proof}

\subsubsection{Proof of Corollary~\ref{lemma:markov}}
\begin{proof}
We have $\vi(t_l) = \diag{\vi(t_{l-1})}\ind{\vK(t_{l-1}) \geq \vec{0}}$.
Thus, Markovianity follows directly from~\eqref{eq:1maturity-dpp} as $(\vK(t_l),\vP(t_l)) = \hat\Psi^T(t_l,\vi(t_l);\vx(t_l)) = \hat\Psi^T(t_l,\diag{\vi(t_{l-1})}\ind{\vK(t_{l-1}) \geq \vec{0}};f(\vx(t_{l-1}),\tilde\epsilon(t_l)))$.
\end{proof}

\subsection{Proofs for the multiple maturity setting of Section~\ref{sec:Kmaturity}}\label{sec:proofs-Kmaturity}
\subsubsection{Proof of Theorem~\ref{thm:Kmaturity-exist}}\label{sec:proof_multi_exist}
\begin{proof}
We will prove the existence of a clearing solution $(\vK^*,\vV^*,\vP^*,\vt^*)$ to~\eqref{eq:Kmaturity} constructively.  Specifically, as in the dynamic programming principle formulation~\eqref{eq:1maturity-dpp} for the single maturity setting, consider the mappings $\hat\Psi$ of the time, solvent institutions, and prior cash account into the current clearing solution.  That is, we define $\hat\Psi$ by
\begin{align*}
&\hat\Psi(t_l,\hat\vK,\hat\vV,\vi) := \left(\begin{array}{c}\hat\Psi_{\vK}(t_l,\hat\vK,\hat\vV,\vi) \\ \hat\Psi_{\vV}(t_l,\hat\vK,\hat\vV,\vi) \\ \hat\Psi_{\vP}(t_l,\hat\vK,\hat\vV,\vi)\end{array}\right) \\ 
    \nonumber &= 
    \FIX_{(\tilde\vK,\tilde\vV,\tilde\vP) \in \bbd(t_l,\hat\vK,\hat\vV,\vi)} \left(\begin{array}{l} 
        \begin{array}{l} \tilde\vV + \beta(\vL^l)^\T\diag{\vi}\ind{\tilde\vK\wedge\tilde\vV < \vec{0}}\\ \qquad + \sum\limits_{k = l+1}^\ell e^{-r(t_k-t_l)} \left[(\vL^k)^\T\diag{\vi}(\beta + (1-\beta)\tilde\vP(t_k)) - \vL^k \vec{1}\right] \end{array} \\ 
        \begin{cases} (I + \diag{\vR(t_l,\va(t_{l-1},\hat\vK,\hat\vV))})\hat\vV + (\vL^l)^\T\diag{\vi}\ind{\tilde\vK\wedge\tilde\vV \geq \vec{0}} - \vL^l\vec{1} &\text{if } l > 0 \\ \vx(0) &\text{if } l = 0 \end{cases} \\ 
        \left(\begin{cases}
            \left[\sum\limits_{\omega_{t_{l+1}} \in \succ(\omega_{t_l})} \frac{\P(\omega_{t_{l+1}})\hat\Psi_{\vP,t_k}(t_{l+1},\vF_l(\tilde\vK,\tilde\vV,\vi)(\omega_{t_l}))(\omega_{t_{l+1}})}{\P(\omega_{t_l})}\right]_{\omega_{t_l} \in \Omega_{t_l}}  &\text{if } l < k \\ 
            \diag{\vi}\ind{\tilde\vK\wedge\tilde\vV \geq \vec{0}} &\text{if } l = k 
        \end{cases}\right)_{k = l}^{\ell}
    \end{array}\right) 
\end{align*}
with \[\vF_l(\tilde\vK,\tilde\vV,\vi) := \left(\tilde\vK \, , \, \tilde\vV + \beta\sum_{k = l}^{\ell} (\vL^k)^\T\diag{\vi}\ind{\tilde\vK\wedge\tilde\vV < \vec{0}} \, , \, \diag{\vi}\ind{\tilde\vK\wedge\tilde\vV \geq \vec{0}}\right)\]
for any $l = 0,1,...,\ell$ and the fixed points taken on the lattice 
\begin{align*}
\bbd(t_l,\hat\vK,\hat\vV,\vi) &:= \bbd_{\vK}(t_l,\hat\vK,\hat\vV,\vi) \times \bbd_{\vV}(t_l,\hat\vK,\hat\vV,\vi) \times [\vec{0},\vec{1}]^{|\Omega_{t_l}| \times (\ell-l+1)} \\
\bbd_{\vK}(t_l,\hat\vK,\hat\vV,\vi) &:= \left\{\vK \in \lcal_{t_l}^n \; \left| \; \vK \in \bbd_{\vV}(t_l,\hat\vK,\hat\vV,\vi) + [\vec{0},\beta(\vL^l)^\T\vi] + \sum_{k = l+1}^\ell e^{-r(t_k-t_l)} ([\beta(\vL^k)^\T\vi,(\vL^k)^\T\vi] - \vL^k\vec{1}) \right.\right\} \\
\bbd_{\vV}(t_l,\hat\vK,\hat\vV,\vi) &:= \begin{cases} 
    \left\{\vV \in \lcal_{t_l}^n \; \left| \; \vV \in (I+\diag{\vR(t_l,\va(t_{l-1},\hat\vK,\hat\vV))})\hat\vV + [\vec{0},(\vL^l)^\T\vi] - \vL^l\vec{1}\right.\right\} &\text{if } l > 0 \\ 
    \{\vx(0)\} &\text{if } l = 0. \end{cases}
\end{align*}
As the problem used to define $\hat\Psi$ is a monotonic mapping on a lattice (proven by inspection for $\hat\Psi_{\vK},\hat\Psi_{\vV}$ and by a simple induction argument for $\hat\Psi_{\vP}$), we can apply Tarski's fixed point theorem to guarantee the existence of a greatest fixed point (i.e., the mapping $\FIX$ is well-defined in this case for every feasible combination of inputs).
Let $(\vK^*,\vV^*,\vP^*)$ be the realized solution from $\hat\Psi(0,\vec{0},\vx(0),\vec{1})$ and define $\vt^*: = \Psi_{\vt}(\vK^*,\vV^*)$.  (We wish to note that the initial value for $\hat\vK$ taken here as $\vec{0}$ is an arbitrary choice as this term is never utilized.)
As in the proof of Proposition~\ref{prop:1maturity-dpp}, we will seek to demonstrate that $(\vK^*,\vV^*,\vP^*,\vt^*)$ is a fixed point to~\eqref{eq:Kmaturity} by proving that $\vP^*(t_l,t_k) = \Psi_{\vP}(t_l,t_k,\vt^*)$ for all times $t_l \leq t_k$.
First, at maturity times, $\vP^*(t_l,t_l) = \prod_{k = 0}^l \ind{\min\{\vK^*(t_l),\vV^*(t_l)\} \geq 0} = \ind{\vt^* > t_l} = \P(\vt^* > t_l | \fcal_{t_l})$ by construction.
Second, consider $l < k$ and assume $\vP^*(t_{l+1},t_k) = \Psi_{\vP}(t_{l+1},t_k,\vt^*)$.  By definition of $\hat\Psi_{\vP,t_k}$, we can construct $\vP^*(t_l,t_k)$ as follows:
\begin{align*}
P_i^*(t_l,t_k,\omega_{t_l}) &= \sum_{\omega_{t_{l+1}} \in \succ(\omega_{t_l})} \frac{\P(\omega_{t_{l+1}}) \hat\Psi_{\vP,t_k,i}(t_{l+1},\vK^*(t_l),\vV^*(t_l)+\beta\sum_{k=l}^\ell(\vL^k)^\T\ind{\vt^* = t_l},\ind{\vt^* > t_l})(\omega_{t_{l+1}})}{\P(\omega_{t_l})} \\
&= \sum_{\omega_{t_{l+1}} \in \succ(\omega_{t_l})} \frac{\P(\omega_{t_{l+1}}) P_i^*(t_{l+1},t_k,\omega_{t_{l+1}})}{\P(\omega_{t_l})} = \sum_{\omega_{t_{l+1}} \in \succ(\omega_{t_l})} \frac{\P(\omega_{t_{l+1}}) \Psi_{\vP,i}(t_{l+1},t_k,\vt^*)(\omega_{t_{l+1}})}{\P(\omega_{t_l})} \\
&= \sum_{\omega_{t_{l+1}} \in \succ(\omega_{t_l})} \frac{\P(\omega_{t_{l+1}}) \P(\tau_i^* > t_k | \fcal_{t_{l+1}})(\omega_{t_{l+1}})}{\P(\omega_{t_l})} = \sum_{\omega_{t_{l+1}} \in \succ(\omega_{t_l})} \frac{\P(\tau_i^* > t_k , \omega_{t_{l+1}})}{\P(\omega_{t_l})} \\
&= \P(\tau_i^* > t_k | \fcal_{t_l})(\omega_{t_l}) = \Psi_{\vP,i}(t_l,t_k,\vt^*)(\omega_{t_l}).
\end{align*}
Therefore, we have constructed a process that clears~\eqref{eq:Kmaturity}. To prove that the enlarged process is Markovian, as with the Corollary~\ref{lemma:markov}, we can simply observe that
\begin{align*}
&(\vK^*(t_l),\vV^*(t_l),\vP^*(t_l,t_k)_{k = l}^\ell) = \hat\Psi(t_l,\vi^*(t_l),\vV^*(t_{l-1});\vx(t_l))  \\ &\qquad \qquad= \hat\Psi(t_l,\diag{\vi^*(t_{l-1})}\ind{\vK^*(t_{l-1})\wedge\vV^*(t_{l-1}) \geq \vec{0}},\vV^*(t_{l-1});f(\vx(t_{l-1}),\tilde\epsilon(t_l))).
\end{align*}
From this, the result is immediate.
\end{proof}

\end{document}